\documentclass[preprint,12pt,authoryear]{elsarticle}

\usepackage{amsmath,amssymb,amsthm,mathtools}
\usepackage{setspace}
\usepackage{hyperref}
\usepackage{enumitem}
\usepackage{booktabs}
\usepackage{xcolor}
\usepackage{graphicx}
\usepackage{float}
\usepackage{threeparttable}
\graphicspath{{figures/}}
\hypersetup{
  colorlinks=true,
  linkcolor=blue, citecolor=blue, urlcolor=blue,
  pdftitle={Identification of Latent Group Effects under Conditional Calibration},
}

\newtheorem{theorem}{Theorem}
\newtheorem{proposition}[theorem]{Proposition}
\newtheorem{lemma}{Lemma}
\newtheorem{corollary}{Corollary}
\theoremstyle{definition}
\newtheorem{assumption}{Assumption}

\theoremstyle{remark}
\newtheorem{remark}{Remark}
\newtheorem{atheorem}{Theorem}[section]
\newtheorem{alemma}[atheorem]{Lemma}
\newtheorem{aproposition}[atheorem]{Proposition}

\newcommand{\E}{\mathbb{E}}
\newcommand{\Pbb}{\mathbb{P}}
\newcommand{\Var}{\operatorname{Var}}
\newcommand{\Cov}{\operatorname{Cov}}
\newcommand{\R}{\mathbb{R}}
\newcommand{\N}{\mathcal{N}}
\newcommand{\1}{\mathbf{1}}
\newcommand{\pto}{\xrightarrow{\;p\;}}
\newcommand{\dto}{\xrightarrow{\;d\;}}
\newcommand{\as}{\text{ a.s.}}
\newcommand{\plim}{\operatorname{plim}}
\newcommand{\sgn}{\operatorname{sgn}}
\newcommand{\Vs}{V^{*}}
\newcommand{\tauor}{\hat{\tau}_{\mathrm{or}}}
\newcommand{\tauorz}{\hat{\tau}^{\,0}_{\mathrm{or}}}
\newcommand{\tauht}{\hat{\Delta}_{\mathrm{ht}}}

\journal{arXiv}

\begin{document}

\begin{frontmatter}

\title{Identification of Latent Group Effects under Conditional
  Calibration}

\author{Marcell T.\ Kurbucz\corref{cor1}}
\ead{m.kurbucz@ucl.ac.uk}
\cortext[cor1]{Corresponding author.}
\affiliation{organization={Institute for Global Prosperity, The Bartlett,
               University College London},
             addressline={9--11 Endsleigh Gardens},
             city={London},
             postcode={WC1H 0EH},
             country={United Kingdom}}

\begin{abstract}
\noindent
\footnotesize
We study identification of a structural group effect when the group
indicator $G\in\{0,1\}$ is unobserved but the analyst observes a calibrated
probability score $p$ satisfying $\E[G\mid p,X]=p$. Under a
constant-coefficient structural mean model, the latent-group coefficient
$\tau$ is point-identified by a closed-form ratio of observable moments
whose denominator is the residual score variance
$V^*=\E[(p-\E[p\mid X])^2]$. Identification fails exactly when the score is
a deterministic function of $X$; we construct an explicit continuum of
observationally equivalent models showing the failure is genuine. The
marginal latent mean gap decomposes as $\tau$ plus a compositional term
that is itself identified in closed form, and we characterise when the two
coincide. The oracle estimator is $\sqrt{n}$-consistent and asymptotically
normal with a closed-form sandwich variance. Under calibration error
bounded by $\delta$, the bias obeys a sharp bound proportional to
$\delta/V^*$, and hard-threshold classification attenuates the estimated
gap. Monte Carlo experiments confirm the theory, including the
variance-weighted estimand under heterogeneous effects.
\end{abstract}

\begin{keyword}
\footnotesize
latent groups \sep identification \sep conditional calibration \sep 
moment equations \sep group effects \\[4pt]
\end{keyword}

\end{frontmatter}

\section{Introduction}
\label{sec:intro}
\noindent
A pervasive challenge in empirical work is the measurement of outcome
differences between groups when group membership is not directly observed.
Poverty status, immigration status, informal employment, fuel insecurity, and
latent health conditions are leading examples. In such settings the analyst
typically has access to a probability score $p_i\in[0,1]$ encoding belief
that unit $i$ belongs to the group of interest, but never observes the binary
indicator $G_i\in\{0,1\}$ itself.

A concrete empirical setting that has all the ingredients is the auditing of
lending disparities when the borrower's race is not recorded. United States
regulators routinely conduct fair lending analysis by proxying the
unrecorded race with the Bayesian Improved Surname Geocoding (BISG)
methodology \citep{cfpb2014}; BISG assigns each applicant a posterior
probability of belonging to each racial group based on surname and geography
\citep{elliott2009}, and the indicator itself is unobserved by both the
lender and the analyst. The audit question---does the expected loan outcome
$Y$ differ by race $G$ \emph{after} conditioning on the credit
characteristics $X$ that legitimately enter the pricing decision?---is
precisely a latent-group-effect question in which the analyst holds a
probability $p$ in place of the indicator, and the proxy probability is built
from information (surname, geography) that is strictly richer than the audit
controls $X$. The algorithmic-fairness literature has documented both the
practice and its pitfalls \citep{chen2019,kallus2022}; what has been missing
is an identification-theoretic account of when the probability itself, used
directly rather than thresholded into a pseudo-indicator, point-identifies
the disparity. This paper provides that account, and its central quantities
are estimable in the BISG setting: the residual score variance below is
computable from $(p,X)$ alone, and self-reported race on a validation
subsample (available, for instance, in mortgage data collected under the Home
Mortgage Disclosure Act) makes the key calibration condition testable.

The central question we address is: \emph{under what conditions, and by what
formula, can a structural group effect be identified from the joint law of
observables $(Y,X,p)$ when $G$ is never observed?}  Throughout, the
structural group effect is the coefficient $\tau$ on the latent indicator in
the conditional mean of the outcome---informally, the shift in the expected
outcome associated with membership for units with the same covariates
(Assumption~\ref{ass:mean} states this formally). We give a precise
answer organised around three claims. First, the structural coefficient $\tau$
is point-identified under mild conditions. Second, identification fails in a
characterisable and sharp way when exactly one of those conditions is violated.
Third, the identified object is distinct from the marginal group mean gap in a
way that can be made fully explicit.

The paper makes four contributions. The first is an identification result.
Under a constant-coefficient structural mean model and the conditional
calibration condition $\E[G\mid p,X]=p$, we prove that $\tau$ is identified
by a weighted moment equation whose denominator
$\Vs=\E[(p-r(X))^{2}]$ is the residual variance of the score after partialling
on $X$. The formula is in closed form and admits a transparent interpretation:
it is formally analogous to an instrumental-variables estimand in which the
score residual $a=p-r(X)$ plays the role of an instrument for the latent
deviation $G-r(X)$; the calibration condition supplies the first-stage
relevance and the mean-independence condition in the structural model supplies
the exclusion restriction.

The second contribution is an exact characterisation of identification
failure. We prove that identification fails exactly at $\Vs=0$, i.e.\ when
the score is a deterministic function of $X$: the moment equation becomes
uninformative there, and, under a mild nondegeneracy condition on the
conditional outcome distribution, we construct an explicit continuum of
observationally equivalent models whose coefficients span a nondegenerate
interval around zero. The construction couples the alternative latent
indicator to the outcome itself; this is essential, because an indicator
generated independently of $Y$ given $(p,X)$ can carry only a zero
structural coefficient. A by-product of the argument is honest partial
identification: at the boundary the observationally equivalent coefficients
form a bounded set governed by the conditional outcome dispersion, not all
of $\R$.

The third contribution is a clean separation between the identified structural
coefficient and the marginal latent mean gap
$\Delta_{\mathrm{marg}}=\E[Y\mid G=1]-\E[Y\mid G=0]$. We decompose
$\Delta_{\mathrm{marg}}=\tau+C$ and show that the compositional term $C$ is
itself point-identified in closed form under conditional calibration---so
the marginal gap is identified alongside $\tau$, and reporting one or the
other is a choice of estimand; we also give a necessary and sufficient
condition for $C=0$.

The fourth contribution is oracle inference and robustness. We establish
$\sqrt{n}$-asymptotic normality of the oracle estimator with an explicit
sandwich variance, compute the exact probability limit under calibration
failure, and derive a sensitivity bound that is sharp over the class of all
calibration error functions bounded uniformly by $\delta$. For feasible
estimation we exhibit a Neyman-orthogonal reformulation of the moment, whose
$\sqrt{n}$-normality under cross-fitting follows by the standard
double-machine-learning argument of \citet{chernozhukov2018}.

Our paper sits at the intersection of four strands of the literature.
Within the \emph{misclassification} literature, \citet{lewbel2007} showed that
average treatment effects are attenuated under misclassification of a binary
regressor and proposed corrections; \citet{mahajan2006} obtained identification
using an instrumental variable; \citet{kasahara2022} extended this to the
endogenous case. Our setting is complementary: instead of observing a noisy
binary label, the analyst observes a calibrated \emph{probability} for $G$,
which changes both the identification argument and the identified object.

The \emph{proxy variable and measurement error} literature
\citep{hu2008,schennach2016} establishes nonparametric identification of full
latent-variable distributions via rank conditions on integral operators. Our
setting is more restrictive---we target only the scalar $\tau$---but our
assumptions are correspondingly weaker and the identification formula is
closed-form. The structure of our moment equation parallels the partially
linear model \citep{robinson1988} and semiparametric IV \citep{newey1990},
and we emphasise that we claim no methodological novelty in the
partialling-out algebra itself. The contribution lies elsewhere, in three
statements that have no counterpart in a partially linear IV regression
with an observable instrument. First, the first stage is \emph{derived},
not assumed: conditional calibration implies that the conditional
covariance between the latent $G$ and the score $p$ equals the conditional
variance of $p$ exactly, so the population first-stage slope is identically
one (Lemma~\ref{lem:cov-gp})---there is no analogue of positing relevance
for an instrument, and no first-stage coefficient to estimate. Second, the
relevance condition $\Vs>0$ and its exact failure boundary are estimable
from $(p,X)$ alone, before the outcome is touched; a first stage involving
a latent regressor admits no such observable diagnostic. Third, at the
boundary the failure is characterised constructively, with the
observationally equivalent coefficients spanning an interval whose width is
governed by the conditional outcome dispersion
(Proposition~\ref{prop:failure}(b))---a partial-identification structure
specific to the latent-indicator model---and away from the boundary the
sharp sensitivity theory prices violations of calibration, the assumption
that replaces instrument validity.
Finally, the \emph{algorithmic fairness} literature \citep{kallus2022,chen2019}
has studied disparity estimation with unobserved protected attributes under
calibration-type assumptions; our contribution to that context is a formal
identification-theoretic treatment with a closed-form formula and exact failure
and sensitivity characterisations.

The remainder of the paper is organised as follows.
Section~\ref{sec:model} sets up the model and discusses how conditionally
calibrated scores arise in practice.
Section~\ref{sec:ident} proves identification and characterises failure.
Section~\ref{sec:structural} distinguishes the structural coefficient from
the marginal gap.
Section~\ref{sec:inference} covers oracle inference.
Section~\ref{sec:robustness} develops robustness to calibration failure.
Section~\ref{sec:feasible} treats feasible estimation, including Neyman
orthogonality.
Section~\ref{sec:mc} presents Monte Carlo evidence.
Section~\ref{sec:discussion} discusses extensions and open directions.
Section~\ref{sec:conclusion} concludes.
The Appendix contains all proofs.

\section{Model and Assumptions}
\label{sec:model}
\noindent
Let $(\Omega,\mathcal{F},\Pbb)$ be a probability space. We observe i.i.d.\
draws $(Y_i,X_i,p_i)\in\R\times\mathcal{X}\times[0,1]$ for $i=1,\ldots,n$,
from the joint distribution $P_{Y,X,p}$. There exists on the same probability
space an unobserved binary variable $G_i\in\{0,1\}$; $G_i=1$ denotes
membership in the latent group of interest. The measurable covariate space
$\mathcal{X}$ is arbitrary.

We write $m(x):=\E[Y\mid X=x]$, $r(x):=\E[p\mid X=x]$, and
$\pi(x):=\E[G\mid X=x]$ for the conditional mean functions. The derived
quantities
\begin{equation}
  z:=2p-1\in[-1,1],\quad R:=Y-m(X),\quad a:=p-r(X)
  \label{eq:derived}
\end{equation}
satisfy $\E[R\mid X]=\E[a\mid X]=0$. The \emph{residual score variance}
\begin{equation}
  \Vs:=\E\!\left[(p-r(X))^{2}\right]=\E[\Var(p\mid X)]\;\geq 0
  \label{eq:Vs}
\end{equation}
measures the variation in $p$ not explained by $X$; it is the key quantity
governing identification.

\begin{assumption}[Structural conditional mean]
\label{ass:mean}
There exist a measurable function $\mu:\mathcal{X}\to\R$ and a scalar
$\tau\in\R$ such that $\E[Y\mid G,p,X]=\mu(X)+\tau G$ a.s.
\end{assumption}

Assumption~\ref{ass:mean} has two components. The effect of latent membership
on the conditional mean of $Y$ is constant in $X$. Additionally, the score
$p$ is mean-independent of $Y$ once $(G,X)$ are known: conditional on true
membership, the analyst's probability score conveys no further information
about the expected outcome.

\begin{assumption}[Conditional calibration]
\label{ass:cal}
$\E[G\mid p,X]=p$ a.s.
\end{assumption}

Assumption~\ref{ass:cal} is the sole formal link between the latent indicator
$G$ and the observed score $p$. It is a calibration condition: $p$ need not
equal the propensity score $\Pbb(G=1\mid X)$ but must be an unbiased predictor
of $G$ given all observed information $(p,X)$. The condition is strictly
stronger than marginal calibration ($\E[G\mid p]=p$): it additionally requires
that the covariates $X$ carry no predictive content for $G$ beyond the score.
This is a substantive restriction, and it is worth stating plainly which
scores satisfy it and which do not. Scores built from information
\emph{coarser} than $X$ generally fail it: an area-level prevalence rate, for
example, is not conditionally calibrated once $X$ contains individual-level
predictors of membership. Conversely, scores built as posterior probabilities
from information \emph{richer} than $X$ satisfy it by construction, as
Remark~\ref{rem:obtain} makes precise. Off-the-shelf classifier outputs sit
in between: even when marginally calibrated they need not be conditionally
calibrated, and the condition should be treated as testable-and-enforceable
rather than automatic. Remark~\ref{rem:obtain} describes how; the sensitivity
theory of Section~\ref{sec:robustness} quantifies the cost of residual
violations.

\begin{remark}[Obtaining conditionally calibrated scores]
\label{rem:obtain}
Since $G$ is unobserved, one may ask how a score satisfying
Assumption~\ref{ass:cal} can ever be constructed or verified---without
observing the very variable whose absence motivates the framework. Two
mechanisms resolve this.

\emph{(i) Calibration by construction.} Let $W$ be an information set with
$\sigma(X)\subseteq\sigma(W)$, and suppose $p=\Pbb(G=1\mid W)$ is a posterior
probability computed from a correctly specified model of $G$ given $W$ (for
example, a Bayes posterior combining a prior with unit-level signals, as in
the BISG construction of \citealp{elliott2009}). Then, because
$\sigma(p,X)\subseteq\sigma(W)$, the tower property gives
\[
  \E[G\mid p,X]=\E\bigl[\E[G\mid W]\mid p,X\bigr]=\E[p\mid p,X]=p ,
\]
so Assumption~\ref{ass:cal} holds \emph{exactly}, at every level of score
informativeness. Conditional calibration is thus not an incidental property
that must be checked score by score; it is the defining property of a
posterior probability based on richer information than the controls.

\emph{(ii) Validation on a labelled subsample.} In the applications that
motivate the framework, $G$ is typically \emph{administratively missing}
rather than unknowable: race is self-reported in some data collections,
poverty status is measured in audit surveys, health conditions are
adjudicated in registry subsamples. On such a labelled subsample the
conditional-calibration condition is directly testable (regress $G-p$ on
$(p,X)$ and test the zero function), and enforceable: recalibrating $p$ on
$(p,X)$ using the labelled units---conditional recalibration---restores the
condition before the score is deployed on the unlabelled bulk where the
inference is run. The companion paper \citep{kurbucz2026kbs} develops this
diagnostic-and-recalibration protocol in detail. Neither mechanism requires
observing $G$ on the analysis sample itself; the framework is therefore not
circular, but it is also not assumption-free---absent both mechanisms,
Assumption~\ref{ass:cal} must be defended on subject-matter grounds and
stress-tested with the sensitivity bound of
Section~\ref{sec:robustness}.
\end{remark}

\begin{remark}[Sharpness versus calibration]
\label{rem:sharp}
Assumption~\ref{ass:var} below requires the score to retain variation beyond
$X$, and the asymptotic variance in Section~\ref{sec:inference} decreases in
$\Vs$: the framework is most informative when the score is \emph{far} from a
deterministic function of $X$. One might worry that this rewards noisy
scores exactly where the calibration assumption is least credible. The
resolution is that two distinct properties are being conflated.
\emph{Sharpness} is how concentrated $p$ is (how close to $\{0,1\}$);
\emph{calibration} is whether $p$ is an unbiased predictor of $G$ given
$(p,X)$. The two are logically independent: a Bayes posterior based on a weak
signal is unsharp but exactly conditionally calibrated (Remark~\ref{rem:obtain}(i)
holds at any signal strength), and the Beta-score design of
Section~\ref{sec:mc} exhibits exact calibration at every noise level
$\sigma_u$. Moreover, $\Vs$ is not noise added to truth: by
Lemma~\ref{lem:cov-gp}, $\Vs=\E[\Cov(G,p\mid X)]$, so residual score
variation \emph{is} the score's information about $G$ beyond $X$---a score
with large $\Vs$ is one that genuinely discriminates within covariate cells.
That said, the concern has a legitimate empirical core: for imperfectly
trained classifiers (as opposed to posterior constructions), weak
discrimination and conditional miscalibration do often co-occur, and the
sensitivity bound of Proposition~\ref{prop:robust} prices exactly this
scenario---the worst-case bias of both oracle estimators in the proposition
is at most $|\tau|\,\delta/(2\Vs)$, because $\E[|z|]\leq 1$ and
$\E[|a|]/\Vs\leq(\Vs)^{-1/2}\leq 1/(2\Vs)$, the last step since
$\Vs\leq\frac14$, so a given
calibration error $\delta$ is \emph{most} damaging when $\Vs$ is small.  In
practice we therefore recommend reporting $\widehat{\Vs}$ together with the
sensitivity bound, and validating calibration on a labelled subsample
whenever one is available (Remark~\ref{rem:obtain}(ii)).
\end{remark}

\begin{assumption}[Non-degenerate residual variation]
\label{ass:var}
$\Vs>0$.
\end{assumption}

\begin{assumption}[Moment conditions]
\label{ass:moments}
$\E[Y^{4}]<\infty$ and $\E[p^{4}]<\infty$.
\end{assumption}

Assumption~\ref{ass:moments} implies square-integrability of all relevant
quantities and is used in Sections~\ref{sec:ident}--\ref{sec:inference}.
Asymptotic normality in Section~\ref{sec:inference} uses the full fourth-moment
condition; identification requires only second moments.

The following two lemmas are the structural backbone of the paper.

\begin{lemma}[Structural decomposition]
\label{lem:basic}
Under Assumptions~\ref{ass:mean}--\ref{ass:moments},
\begin{equation}
  \pi(X)=r(X)\as, \qquad m(X)=\mu(X)+\tau r(X)\as
  \label{eq:struct-decomp}
\end{equation}
\end{lemma}

\begin{lemma}[Residual decomposition]
\label{lem:residual}
Under Assumptions~\ref{ass:mean}--\ref{ass:moments}, the outcome residual
satisfies
\begin{equation}
  R=\tau(G-r(X))+\varepsilon,\qquad \E[\varepsilon\mid G,p,X]=0.
  \label{eq:resid-decomp}
\end{equation}
\end{lemma}

Proofs are in Appendix~\ref{app:proofs}. The content of
Lemma~\ref{lem:residual} is structural: the entire predictable part of the
outcome residual $R$ is driven by the deviation of true group membership from
its score-implied expectation, $G-r(X)$.

\section{Identification}
\label{sec:ident}
\noindent
\begin{theorem}[Population moment identity and point identification]
\label{thm:moment}
\leavevmode
\begin{enumerate}[label=\textup{(\alph*)}]
  \item Under Assumptions~\ref{ass:mean}, \ref{ass:cal},
    and~\ref{ass:moments},
\begin{equation}
  \E\!\left[(2p-1)(Y-m(X))\right]=2\tau\Vs.
  \label{eq:moment}
\end{equation}
  \item If in addition Assumption~\ref{ass:var} holds, then
\begin{equation}
  \tau=\frac{\E[(2p-1)(Y-m(X))]}{2\,\Vs},
  \label{eq:tau-id}
\end{equation}
so the structural coefficient $\tau$ is point-identified from the joint law
of $(Y,X,p)$.
\end{enumerate}
\end{theorem}

The proof of Theorem~\ref{thm:moment} is in Appendix~\ref{app:thm-moment}.
The identification formula~\eqref{eq:tau-id} has a transparent algebraic
structure. The numerator $\E[zR]$ is the covariance between the signed score
$z=2p-1$ and the outcome residual $R$, after partialling both on $X$. The
denominator $2\Vs=2\E[a^{2}]$ is twice the residual variance of the score.
The ratio is therefore the slope of the regression of $R$ on $z$ in the
covariate-partialled data, which by Lemma~\ref{lem:cov-gp} in the Appendix
equals the slope of the regression of $R$ on the latent deviation $G-r(X)$.
This is formally analogous to an IV estimand in which $a=p-r(X)$ acts as an
instrument for $G-r(X)$: the calibration condition (Assumption~\ref{ass:cal})
supplies the first-stage relevance, and the mean-independence condition in
Assumption~\ref{ass:mean} supplies the exclusion restriction.

\begin{remark}[Choice of instrument]
\label{rem:instrument}
The signed score $z=2p-1$ is a presentational choice, not a substantive one.
Using $p$ itself in the numerator yields, by the same argument, the simpler
moment condition
\[
  \E[p\,(Y-m(X))]=\tau\Vs ,
\]
which identifies the same $\tau$: writing $p=a+r(X)$ and using
$\E[a\mid X]=0$ collapses the calibration step to
$\E[p(p-r(X))]=\E[a^{2}]=\Vs$ directly. Indeed, any instrument of the form
$\alpha+\beta p$ with $\beta\neq 0$ identifies the same ratio. We retain
$z=2p-1$ because its symmetric range $[-1,1]$ gives the moment equation a
signed-score interpretation that carries over to the hard-threshold analysis
(the attenuation factor $\kappa=2\,\E[|p-\tfrac12|]=\E[|z|]$ of
Appendix~\ref{app:threshold} and the sensitivity bound of
Proposition~\ref{prop:robust}(b) are both naturally expressed in $z$), and
because the companion applied paper \citep{kurbucz2026kbs} states
its inherited results in this form. Readers who prefer the minimal
presentation may substitute $p$ for $z$ throughout at the cost of carrying
$\Vs$ in place of $2\Vs$.
\end{remark}

We now show that Assumption~\ref{ass:var} is not merely a regularity
condition but the exact boundary of identification.

\begin{proposition}[Identification failure]
\label{prop:failure}
\leavevmode
\begin{enumerate}[label=\textup{(\alph*)}]
  \item If $\Vs=0$, then both sides of~\eqref{eq:moment} equal zero for
    every $\tau\in\R$.
  \item Suppose $\Vs=0$ and, in addition, there exist constants
    $\underline r\in(0,\tfrac12]$ and $\underline\kappa>0$ such that
    $\underline r\leq r(X)\leq 1-\underline r$ a.s.\ and
    $\E\bigl[(Y-m(X))\,g(Y-m(X))\mid X\bigr]\geq\underline\kappa$ a.s.,
    where $g(u):=u/(1+|u|)$. Set $\bar\tau:=2\,\underline r\,\underline\kappa$.
    Then for every $\tau'$ with $|\tau'|\leq\bar\tau$ there exists a model
    satisfying Assumptions~\ref{ass:mean}, \ref{ass:cal},
    and~\ref{ass:moments} with latent-group coefficient $\tau'$ in which
    the observables $(Y,X,p)$ have exactly their original joint
    distribution. Hence $\tau$ is not point-identified: the identified set
    contains the nondegenerate interval $[-\bar\tau,\bar\tau]$.
  \item $\Vs=0$ if and only if $p=r(X)$ almost surely.
\end{enumerate}
\end{proposition}

Part (b) is the substantive non-identification claim, and the form of its
construction matters. A tempting shortcut---generating an alternative
indicator $G'$ from the score by independent randomisation and then
postulating the structural equation---fails: an indicator independent of
$Y$ given $(p,X)$ satisfies $\E[Y\mid G',p,X]=\E[Y\mid p,X]$, which cannot
depend on $G'$, so such a $G'$ can carry only the coefficient $\tau'=0$. A
valid alternative model must couple $G'$ to the outcome. The construction
in Appendix~\ref{app:prop-failure} does so by tilting the conditional
membership probability with a bounded, conditionally mean-zero transform of
the outcome residual: $G'$ is drawn with
$\Pbb(G'=1\mid Y,X)=r(X)+c(X)\,\psi(Y,X)$, where $\E[\psi\mid X]=0$
preserves conditional calibration (recall $p=r(X)$ when $\Vs=0$) and the
positive conditional covariance between $\psi$ and $Y$ generates exactly
the group gap $\tau'$, with $\mu'(X)=m(X)-\tau'r(X)$. The observables are
untouched, so observational equivalence is exact. The interval restriction
$|\tau'|\leq\bar\tau$ is not an artifact of the method: a latent binary
split of a fixed conditional outcome distribution can only support a mean
gap commensurate with that distribution's dispersion, so at the boundary
the identified set is a bounded set rather than all of $\R$; in degenerate
cases---for instance $\Var(Y\mid X)=0$ a.s.\ with $r(X)$ interior---it
collapses to $\{0\}$, and the nondegeneracy condition in part (b) is what
rules this out. Both conditions in part (b)---strict overlap of $r(X)$ and
the dispersion bound on
$\E[(Y-m(X))\,g(Y-m(X))\mid X]$---are restrictions on the observable joint
law of $(Y,X)$, so they are in principle verifiable from the data whose
identification content is being characterised.

\begin{remark}[Heterogeneous effects]
\label{rem:hetero}
The constant-coefficient restriction in Assumption~\ref{ass:mean} can be
relaxed without changing the argument. If
$\E[Y\mid G,p,X]=\mu(X)+\tau(X)\,G$ for a measurable function
$\tau:\mathcal{X}\to\R$ with $\E[\tau(X)^{2}]<\infty$, the same proof shows
that the moment equation~\eqref{eq:moment} holds with $\tau\Vs$ replaced by
$\E[\tau(X)\Var(p\mid X)]$, so the identified estimand becomes the
variance-weighted average
\[
  \bar\tau
  \;=\;
  \frac{\E\bigl[\tau(X)\,\Var(p\mid X)\bigr]}{\E\bigl[\Var(p\mid X)\bigr]},
\]
which reduces to $\tau$ under the constant-coefficient restriction. The
weight $\Var(p\mid X)$ is the local informativeness of the score at
covariate value $X$: cells where the score genuinely discriminates
contribute more. This is the exact analogue of the variance-weighting
familiar from linear IV and partially linear models under effect
heterogeneity, and the analyst should interpret the estimand accordingly
when constancy is implausible. Section~\ref{sec:mc-hetero} verifies the
weighting formula in simulation, including a design in which
$\bar\tau$ differs from $\E[\tau(X)]$ by construction.
\end{remark}

\section{The Structural Coefficient and the Marginal Gap}
\label{sec:structural}
\noindent
A natural question is whether $\tau$ equals the marginal latent mean gap
$\Delta_{\mathrm{marg}}:=\E[Y\mid G=1]-\E[Y\mid G=0]$. Under
Assumption~\ref{ass:mean}, a direct calculation gives
\begin{equation}
  \Delta_{\mathrm{marg}}=\tau+C,\qquad
  C:=\E[\mu(X)\mid G=1]-\E[\mu(X)\mid G=0].
  \label{eq:marg-decomp}
\end{equation}
The term $C$ captures differences in covariate composition across latent
groups. Although it depends on the latent conditional distributions
$\Pbb_{X\mid G=g}$, it is nevertheless point-identified: Assumption~\ref{ass:cal}
and the tower property give $\E[\mu(X)G]=\E[\mu(X)\,p]$ and
$\E[G]=\E[p]$, so
\[
  C=\frac{\E[\mu(X)\,p]}{\E[p]}-\frac{\E[\mu(X)(1-p)]}{\E[1-p]},
  \qquad \mu(X)=m(X)-\tau\,r(X),
\]
where $\mu$ is identified once $\tau$ is (Theorem~\ref{thm:moment}). The
same argument identifies the full latent-group covariate law,
$\Pbb(X\in A\mid G=1)=\E[\mathbf{1}_A(X)\,p]/\E[p]$: conditional
calibration pins down every moment of $X$ within each latent group.

\begin{corollary}[Structural coefficient versus marginal gap]
\label{cor:sv-marginal}
Under Assumptions~\ref{ass:mean}--\ref{ass:var}, the following are
equivalent: (i)~$\Delta_{\mathrm{marg}}=\tau$; (ii)~$C=0$;
(iii)~the latent groups are covariate-balanced,
$\E[\mu(X)\mid G=1]=\E[\mu(X)\mid G=0]$.
\end{corollary}

The proof is immediate from~\eqref{eq:marg-decomp}. The practical import is
that $\tau$ identifies the \emph{within-covariate-cell} group effect, while
$\Delta_{\mathrm{marg}}$ adds the compositional term $C$. Since $C$ is
identified in closed form by the display above, both estimands are
available from $(Y,X,p)$, and reporting $\tau$ or $\Delta_{\mathrm{marg}}$
is a choice of estimand rather than an identification issue; the
decomposition makes explicit what each one measures.

\section{Oracle Estimation and Inference}
\label{sec:inference}
\label{sec:oracle}
\noindent
Suppose for this section that $m$ and $r$ are known; feasible estimation
with estimated nuisance functions is treated in Section~\ref{sec:feasible}. Define the
\emph{oracle estimator}
\begin{equation}
  \tauor:=\frac{\frac{1}{n}\sum_{i}(2p_i-1)(Y_i-m(X_i))}
               {2\,\frac{1}{n}\sum_{i}(p_i-r(X_i))^{2}},
  \label{eq:oracle}
\end{equation}
and the score evaluated at the true parameter,
\begin{equation}
  \psi_i:=(2p_i-1)(Y_i-m(X_i))-2\tau(p_i-r(X_i))^{2}.
  \label{eq:score}
\end{equation}
By Theorem~\ref{thm:moment}, $\E[\psi_i]=0$.

\begin{theorem}[Oracle CLT]
\label{thm:clt}
Under Assumptions~\ref{ass:mean}--\ref{ass:moments} and i.i.d.\ sampling,
\begin{equation}
  \sqrt{n}\,(\tauor-\tau)\;\dto\;\N(0,\,\sigma^{2}_{\mathrm{or}}),
  \qquad
  \sigma^{2}_{\mathrm{or}}
  =\frac{\E[\psi_i^{2}]}{(2\Vs)^{2}}.
  \label{eq:clt}
\end{equation}
\end{theorem}

The variance $\sigma^{2}_{\mathrm{or}}=J^{-2}\E[\psi_i^{2}]$ has the
standard sandwich form with Jacobian $J=2\Vs$. The identification
condition $\Vs>0$ is precisely the condition that $J\neq 0$, i.e., that the
moment equation is locally informative about $\tau$ in a neighbourhood of the
truth.

\begin{corollary}[Consistent variance estimator and Wald interval]
\label{cor:wald}
Let $\hat\psi_i:=(2p_i-1)(Y_i-m(X_i))-2\tauor(p_i-r(X_i))^{2}$. The
estimator
\begin{equation}
  \hat\sigma^{2}_{\mathrm{or}}
  :=\frac{\frac{1}{n}\sum_{i}\hat\psi_i^{2}}
         {\Bigl(2\,\frac{1}{n}\sum_{i}(p_i-r(X_i))^{2}\Bigr)^{2}}
  \label{eq:varhat}
\end{equation}
satisfies $\hat\sigma^{2}_{\mathrm{or}}\pto\sigma^{2}_{\mathrm{or}}$, and
$\tauor\pm z_{1-\alpha/2}\,\hat\sigma_{\mathrm{or}}/\sqrt{n}$ has asymptotic
coverage $1-\alpha$.
\end{corollary}

Proofs of both results are in Appendix~\ref{app:inference}. The proof of the
CLT proceeds by applying the delta method to $f(u,v)=u/(2v)$ after the
bivariate CLT for $(\bar U_n,\bar V_n)$; the cross-terms in the delta-method
expansion cancel when expressed in terms of the centred score $\psi_i$, leaving
the clean formula~\eqref{eq:clt}.

\section{Robustness to Calibration Failure}
\label{sec:robustness}
\noindent
Suppose Assumption~\ref{ass:cal} is violated and
\begin{equation}
  \E[G\mid p,X]=p+\eta(p,X)
  \label{eq:miscal-model}
\end{equation}
for a measurable calibration error function $\eta$. The next result gives
the exact probability limit of the oracle estimator under
\eqref{eq:miscal-model} and the sharp worst-case bias over all admissible
calibration error functions of a given magnitude, where admissibility
requires $p+\eta$ to remain a valid conditional probability.

\begin{proposition}[Bias under calibration failure and sharp sensitivity
  bound]
\label{prop:robust}
Suppose Assumption~\ref{ass:mean} holds, \eqref{eq:miscal-model} holds, and
$\Vs>0$. Let $m_0(X):=\mu(X)+\tau r(X)$ be the conditional mean of $Y$ given
$X$ under calibration (Lemma~\ref{lem:basic}), and let $\tauorz$ denote the
estimator~\eqref{eq:oracle} with $m_0$ in place of $m$.
\begin{enumerate}[label=\textup{(\alph*)}]
  \item Then
\begin{equation}
  \plim\;\tauorz=\tau+B_{\mathrm{cal}},\qquad
  B_{\mathrm{cal}}:=\frac{\tau\,\E[z\,\eta(p,X)]}{2\Vs}.
  \label{eq:bias}
\end{equation}
  \item Fix $\delta\in[0,\tfrac12]$ and let
    $\mathcal{H}_\delta:=\{\eta:\;|\eta|\leq\delta\;\text{and}\;
    0\leq p+\eta\leq 1\;\text{a.s.}\}$ be the class of admissible
    calibration errors of magnitude $\delta$ (admissibility ensures that
    $p+\eta$ is a valid conditional probability). Then
\begin{equation}
  \sup_{\eta\in\mathcal{H}_\delta}
  \left|\plim\;\tauorz-\tau\right|
  =|\tau|\cdot\frac{\delta\,\E[|z|]}{2\Vs},
  \label{eq:sensitivity}
\end{equation}
    and the supremum is attained at $\eta^{*}(p,X)=-\delta\,\sgn(z)$.
  \item Under \eqref{eq:miscal-model},
    $m(X)=\E[Y\mid X]=m_0(X)+\tau\bar\eta(X)$ with
    $\bar\eta(X):=\E[\eta(p,X)\mid X]$, and the oracle
    estimator~\eqref{eq:oracle} satisfies
\[
  \plim\;\tauor=\tau+\frac{\tau\,\E[a\,\eta(p,X)]}{\Vs},
  \qquad
  \sup_{\eta\in\mathcal{H}_\delta}
  \left|\plim\;\tauor-\tau\right|
  \leq|\tau|\cdot\frac{\delta\,\E[|a|]}{\Vs},
\]
    with equality in the bound, attained at $\eta(p,X)=-\delta\,\sgn(a)$,
    when $\min(p,1-p)\geq\delta$ a.s. The limits of $\tauor$ and $\tauorz$
    coincide when $\E[(2r(X)-1)\,\eta(p,X)]=0$, in particular when
    $\E[\eta\mid X]=0$ a.s.\ or $r(X)=\frac12$ a.s.
\end{enumerate}
\end{proposition}

Parts (a) and (b) describe the oracle that subtracts the calibrated-model
mean $m_0$, as in the proof of Theorem~\ref{thm:moment} and in the
experiments of Section~\ref{sec:mc-calibration}.
The bias $B_{\mathrm{cal}}$ is proportional to $\tau$: no bias arises when
the true effect is zero, regardless of miscalibration. It is proportional to
the score-weighted mean of the calibration error; it vanishes whenever
$\E[z\eta]=0$, which holds when the miscalibration is symmetric
in the sense of being orthogonal to the signed score.

Part (c) describes the estimator that subtracts $\E[Y\mid X]$. Under
miscalibration this mean absorbs the part of $\eta$ that depends on $X$
alone, and the residualised score $a$ takes the place of $z/2$ in the bias.
The bias therefore vanishes whenever $\E[a\eta]=0$, in particular whenever
$\eta$ depends on $X$ only, since then
$\E[a\eta]=\E[\eta(X)\,\E[a\mid X]]=0$. The feasible
estimators~\eqref{eq:plugin} and~\eqref{eq:ort-estimator} have the same
limit when $\hat m$ and $\hat r$ are consistent for $m$ and $r$, because the
outcome regression then targets $\E[Y\mid X]$.

The bound in~\eqref{eq:sensitivity} has a signal-to-noise
interpretation. The denominator $2\Vs/\E[|z|]$ is an effective informativeness
measure of the score; larger $\Vs$ means the score is more discriminating, and
the same calibration error $\delta$ produces proportionally less bias. As
$\Vs\to 0$, the bound diverges, consistently with the identification failure of
Proposition~\ref{prop:failure}; conversely, small-$\Vs$ designs are exactly
where calibration error is most damaging, the quantitative content of the
sharpness-versus-calibration discussion in Remark~\ref{rem:sharp}. The bound
in part (c) is at most $|\tau|\,\delta/(\Vs)^{1/2}$, because
$\E[|a|]\leq(\Vs)^{1/2}$, and the same reading applies. Proofs
are in Appendix~\ref{app:robust}.

\section{Feasible Estimation and Neyman Orthogonality}
\label{sec:feasible}
\label{sec:plugin}
\noindent
When $m$ and $r$ are unknown, replace them with estimators $\hat m$ and
$\hat r$ to obtain
\begin{equation}
  \hat\tau:=\frac{\frac{1}{n}\sum_{i}(2p_i-1)(Y_i-\hat m(X_i))}
                 {2\,\frac{1}{n}\sum_{i}(p_i-\hat r(X_i))^{2}}.
  \label{eq:plugin}
\end{equation}
The denominator stability under nuisance estimation error is non-trivial and
is isolated as a separate lemma.

\begin{lemma}[Denominator stability]
\label{lem:denom}
If $p_i,\hat r(X_i)\in[0,1]$ a.s.\ and the estimated regression
satisfies $n^{-1}\sum_i(\hat r(X_i)-r(X_i))^{2}\pto 0$, then
\[
  n^{-1}\textstyle\sum_i(p_i-\hat r(X_i))^{2}
  -n^{-1}\textstyle\sum_i(p_i-r(X_i))^{2}\pto 0.
\]
\end{lemma}

\begin{proposition}[Plug-in consistency]
\label{prop:plugin}
Under Assumptions~\ref{ass:mean}--\ref{ass:moments}, i.i.d.\ sampling,
$p_i,\hat r(X_i)\in[0,1]$ a.s., and
\[
  n^{-1}\textstyle\sum_i(\hat m(X_i)-m(X_i))^{2}\pto 0,\qquad
  n^{-1}\textstyle\sum_i(\hat r(X_i)-r(X_i))^{2}\pto 0,
\]
we have $\hat\tau\pto\tau$.
\end{proposition}

Proofs are in Appendix~\ref{app:plugin}.

For $\sqrt{n}$-normality of $\hat\tau$ with nuisances estimated at nonparametric
rates, the score~\eqref{eq:score} must be Neyman-orthogonal
\citep{chernozhukov2018}. Appendix~\ref{app:ortho} verifies that the
$r$-Gateaux derivative of $\E[\psi]$ is already zero (because $\E[a\mid X]=0$),
while the $m$-Gateaux derivative equals $-\E[(2r(X)-1)\delta_m(X)]$, which is
non-zero whenever $r(X)\not\equiv\frac12$. The score therefore fails Neyman
orthogonality through its $m$-direction.

Appendix~\ref{app:ortho} also identifies a natural Neyman-orthogonal
reformulation. Replacing $(2p-1)$ by $2(p-r(X))=2a$ in the numerator gives
the score
\begin{equation}
  \tilde\psi_i:=2\,(p_i-r(X_i))\bigl(Y_i-m(X_i)-\tau\,(p_i-r(X_i))\bigr),
  \label{eq:ort-score}
\end{equation}
which has both Gateaux derivatives equal to zero. The estimator defined by
solving $\frac{1}{n}\sum_i\tilde\psi_i(\hat\tau)=0$ is
\begin{equation}
  \hat\tau_{\mathrm{ort}} :=
  \frac{\frac{1}{n}\sum_i(p_i-\hat r(X_i))(Y_i-\hat m(X_i))}
       {\frac{1}{n}\sum_i(p_i-\hat r(X_i))^2},
  \label{eq:ort-estimator}
\end{equation}
which is a distinct estimator from~\eqref{eq:plugin}. When nuisances are
known, both estimators converge to $\tau$ and are asymptotically equivalent;
with estimated nuisances, $\hat\tau_{\mathrm{ort}}$ is the natural candidate
for DML-compatible inference.

Formal inference for the cross-fitted estimator based on
\eqref{eq:ort-score} follows from the double-machine-learning framework of
\citet{chernozhukov2018}: the score is linear in $\tau$,
Appendix~\ref{app:ortho} proves its Neyman orthogonality in both nuisance
directions, and under the standard rate conditions on the nuisance
estimators (each consistent in $L_2(P)$ at rate $o_P(n^{-1/4})$, with
cross-fitting) the generic DML central limit theorem applies to
$\hat\tau_{\mathrm{ort}}$ in~\eqref{eq:ort-estimator}, delivering
$\sqrt{n}$-normality with the sandwich variance evaluated at the orthogonal
score. We do not restate that argument; the Monte Carlo coverage results
for the orthogonal estimator reported in Section~\ref{sec:mc} are
consistent with the theory rather than purely exploratory. Two caveats keep
the scope honest: this result, like Theorem~\ref{thm:clt}, treats the
score-generating mechanism as externally supplied and conditions on it, and
it delivers no uniformity over $\Vs\to 0$; inference local to the
identification boundary is an open problem (Section~\ref{sec:discussion}).

\section{Monte Carlo Evidence}
\label{sec:mc}
\noindent
We report five sets of simulations, each tied directly to a theoretical result.
All experiments use $R=2{,}000$ replications with seeded random draws.
The baseline DGP has $X\in\R^3$ with independent standard normal entries,
$r(X)=\sigma(\beta_r^\top X)$ (logistic), and structural baseline
$\mu(X)=\beta_m^\top X$ (linear).
The score is drawn as $p\mid X\sim\mathrm{Beta}(r(X)\,c,\,(1-r(X))\,c)$
with concentration $c=(1-\sigma_u^2)/\sigma_u^2$ (the symbol $\kappa$ is
reserved for the attenuation factor), so that $\E[p\mid X]=r(X)$ and
$\Var(p\mid X)=\sigma_u^2\,r(X)(1-r(X))$ exactly, giving
$\Vs = \sigma_u^2\,\E[r(X)(1-r(X))]$ exactly (not an approximation).
The outcome is $Y=\mu(X)+\tau G+\varepsilon$ with
$G\sim\mathrm{Bernoulli}(p)$ and $\varepsilon\sim N(0,1)$; by
Lemma~\ref{lem:basic} the induced conditional mean is
$m(X)=\E[Y\mid X]=\mu(X)+\tau r(X)$, and it is this $m(X)$---not
$\mu(X)$---that the oracle estimator subtracts (the replication code
constructs $m$ accordingly). Because $G$ is drawn from its score,
$\E[G\mid p,X]=p$ holds \emph{exactly} at every noise level $\sigma_u$: the
design realises the sharpness--calibration separation of
Remark~\ref{rem:sharp}, with $\sigma_u$ moving sharpness (and $\Vs$) while
calibration stays exact.
The oracle estimator uses the true nuisance functions $m(X)$ and $r(X)$;
the plug-in estimator fits degree-2 polynomial ridge regressions without
cross-fitting; the orthogonal estimator uses 5-fold cross-fitting with the
same ridge models; and the hard-threshold estimator replaces $p$ with
$\1\{p>\frac12\}$.
Full replication code is provided in the online supplement.

\subsection{Finite-sample performance and oracle normality}
\label{sec:mc-baseline}

\begin{sloppypar}
\noindent
Table~\ref{tab:sim1} reports bias, standard deviation, RMSE, and empirical
coverage of nominal 95\% Wald intervals for the three main estimators at
$n\in\{500,1{,}000,5{,}000\}$ with $\tau=1$ and $\sigma_u=0.30$. The oracle
estimator is approximately unbiased throughout; the plug-in estimator
exhibits a persistent positive bias of roughly 0.12--0.17, attributable to
regularisation bias of the in-sample ridge fit: the shrunken $\hat m$
under-partials the score-dependent component $\tau\,r(X)$ of
$m(X)=\mu(X)+\tau r(X)$ (Lemma~\ref{lem:basic}), and the leftover
$\E[z\,(m-\hat m)]>0$ inflates the numerator; the orthogonal
estimator, which uses 5-fold cross-fitting, is nearly unbiased and achieves
coverage close to the nominal 0.95. RMSE shrinks at the $\sqrt{n}$ rate for
all three estimators.
\end{sloppypar}

Two features of Table~\ref{tab:sim1} deserve comment, since at first sight
they may look anomalous: at $n=500$ and $n=1{,}000$ the biased plug-in
estimator has \emph{smaller} RMSE than the unbiased oracle, and the
orthogonal estimator's standard deviation is smaller than the oracle's at
every $n$. Neither is an error, and neither contradicts the theory, because
the oracle is efficient only within the class of estimators built on
\emph{its own score}: no result here (or in general) implies that an
estimator using the true nuisance functions dominates one using estimated
nuisances when the scores differ---indeed the analogous phenomenon, that
estimated nuisances can strictly reduce asymptotic variance relative to
known ones, is classical \citep{hirano2003}. Mechanically, the decomposition
$\mathrm{RMSE}^{2}=\mathrm{bias}^{2}+\mathrm{SD}^{2}$ in the table shows what
happens: at $n=500$ the plug-in trades a bias of $0.170$
($\mathrm{bias}^{2}=0.029$) for a variance reduction from $0.501^{2}=0.251$
to $0.359^{2}=0.129$, a favourable trade that reverses by $n=5{,}000$ as
variance shrinks and the non-vanishing bias comes to dominate. The variance
gap itself has a structural source: writing $z=2a+(2r(X)-1)$, the oracle
score decomposes as $\psi=\tilde\psi+(2r(X)-1)R$, where $\tilde\psi$ is the
orthogonal score~\eqref{eq:ort-score}; the extra term $(2r(X)-1)R$ is pure
noise (it has mean zero conditional on $X$ and is essentially uncorrelated
with $\tilde\psi$) and accounts for more than half of $\Var(\psi)$ in this
design. Estimators that residualise on $X$---explicitly, as the orthogonal
estimator does, or implicitly, as the plug-in's fitted $\hat m$ does
in-sample---strip most of this component out. Evaluating the two asymptotic
standard deviations $\smash{\sqrt{\Var(\psi)}/(2\Vs)}$ and
$\smash{\sqrt{\Var(\tilde\psi)}/(2\Vs)}$ under the DGP reproduces the
oracle and orthogonal columns of Table~\ref{tab:sim1} to within Monte Carlo
error ($0.508$ vs.\ $0.501$ and $0.343$ vs.\ $0.329$ at $n=500$, and
analogously at larger $n$), confirming that the ordering is structural
rather than a finite-sample artifact.

Figure~\ref{fig:qq} shows normal QQ-plots of the standardised oracle estimates
$\sqrt{n}(\hat\tau_{\mathrm{or}}-\tau)/\hat\sigma_{\mathrm{or}}$ at each
sample size. The agreement with the $N(0,1)$ reference is excellent at
$n=1{,}000$ and $n=5{,}000$, confirming Theorem~\ref{thm:clt}.

\begin{table}[H]
\centering
\caption{Finite-sample performance under correct specification ($\tau=1$,
$\sigma_u=0.30$, $R=2{,}000$ replications). The orthogonal estimator uses
5-fold cross-fitting; the plug-in estimator fits its ridge nuisances
in-sample, and its positive bias is the resulting regularisation bias. The
plug-in's smaller RMSE at $n\leq 1{,}000$ and the orthogonal estimator's
uniformly smaller SD are structural, not anomalous; see the discussion
following the table.}
\label{tab:sim1}
\begin{threeparttable}
\begin{tabular}{llrrrr}
\toprule
Estimator & $n$ & Bias & SD & RMSE & Coverage \\ \midrule
Oracle     & 500   & $+$0.004 & 0.501 & 0.501 & 0.952 \\
Plug-in    & 500   & $+$0.170 & 0.359 & 0.397 & 0.987 \\
Orthogonal & 500   & $-$0.014 & 0.329 & 0.329 & 0.955 \\ \midrule
Oracle     & 1,000 & $-$0.016 & 0.365 & 0.365 & 0.945 \\
Plug-in    & 1,000 & $+$0.150 & 0.252 & 0.294 & 0.983 \\
Orthogonal & 1,000 & $-$0.010 & 0.233 & 0.233 & 0.958 \\ \midrule
Oracle     & 5,000 & $-$0.005 & 0.162 & 0.163 & 0.946 \\
Plug-in    & 5,000 & $+$0.119 & 0.114 & 0.165 & 0.956 \\
Orthogonal & 5,000 & $-$0.005 & 0.110 & 0.110 & 0.946 \\
\bottomrule
\end{tabular}
\begin{tablenotes}\small
\item \textit{Notes}: Coverage is the empirical frequency of nominal 95\%
Wald intervals over $R=2{,}000$ replications.
\end{tablenotes}
\end{threeparttable}
\end{table}

\begin{figure}[H]
\centering
\includegraphics[width=\textwidth]{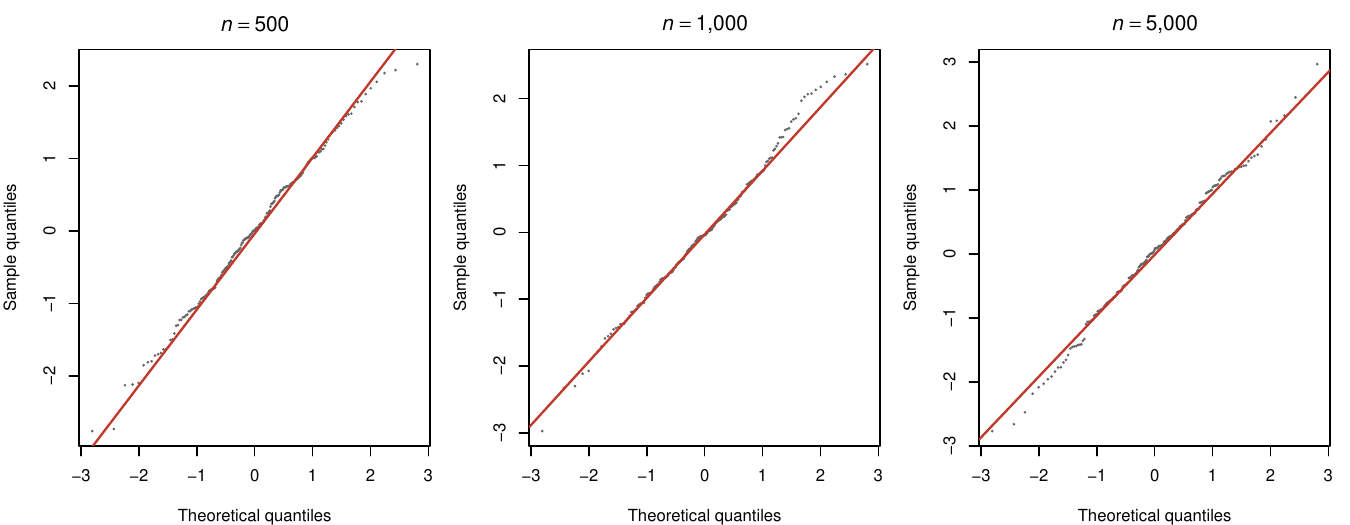}
\caption{Normal QQ-plots of standardised oracle estimates
$\sqrt{n}(\hat\tau_{\mathrm{or}}-\tau)/\hat\sigma_{\mathrm{or}}$
at $n\in\{500,1{,}000,5{,}000\}$. Confirming Theorem~\ref{thm:clt}.}
\label{fig:qq}
\end{figure}

\subsection{Approach to the identification boundary}
\label{sec:mc-boundary}

\noindent
Proposition~\ref{prop:failure} predicts that $\tau$ is not identified when
$\Vs=0$, and Theorem~\ref{thm:clt} that the RMSE grows as $1/\Vs$ as
$\Vs\to 0$ whenever $\E[\psi_i^2]$ stays bounded away from zero, as it does
in this design.
Table~\ref{tab:sim2} traces this by decreasing score noise $\sigma_u$ at
$n=1{,}000$. As $\Vs$ falls from $5.6\times 10^{-2}$ to $2.3\times 10^{-7}$,
RMSE grows by five orders of magnitude, while coverage remains close to
its nominal level throughout---the widening confidence intervals correctly
track the growing variance. Figure~\ref{fig:boundary} plots RMSE on a
log-log scale (left) and CI coverage (right); the empirical RMSE tracks
the theoretical $\mathrm{RMSE}\propto 1/\Vs$ reference closely.

\begin{table}[H]
\centering
\caption{Approach to the identification boundary ($n=1{,}000$, $\tau=1$,
$R=2{,}000$ replications). Oracle estimator.}
\label{tab:sim2}
\begin{threeparttable}
\begin{tabular}{rrrrrr}
\toprule
$\sigma_u$ & True $\Vs$ & Bias & SD & RMSE & Coverage \\ \midrule
0.500 & $5.63\times 10^{-2}$ & $-$0.001 & 0.164 & 0.164 & 0.957 \\
0.250 & $1.41\times 10^{-2}$ & $+$0.003 & 0.475 & 0.475 & 0.956 \\
0.100 & $2.25\times 10^{-3}$ & $+$0.044 & 2.578 & 2.578 & 0.953 \\
0.050 & $5.63\times 10^{-4}$ & $-$0.318 & 9.559 & 9.562 & 0.952 \\
0.010 & $2.25\times 10^{-5}$ & $+$6.363 & 242.8 & 242.8 & 0.946 \\
0.005 & $5.61\times 10^{-6}$ & $+$17.11 & 978.3 & 978.2 & 0.952 \\
0.001 & $2.25\times 10^{-7}$ & $-$158.3 & 24858 & 24853 & 0.947 \\
\bottomrule
\end{tabular}
\begin{tablenotes}\small
\item \textit{Notes}: $\Vs=\E[(p-r(X))^2]$. SD, RMSE, and coverage are
computed on finite estimates only.
\end{tablenotes}
\end{threeparttable}
\end{table}

\begin{figure}[H]
\centering
\includegraphics[width=\textwidth]{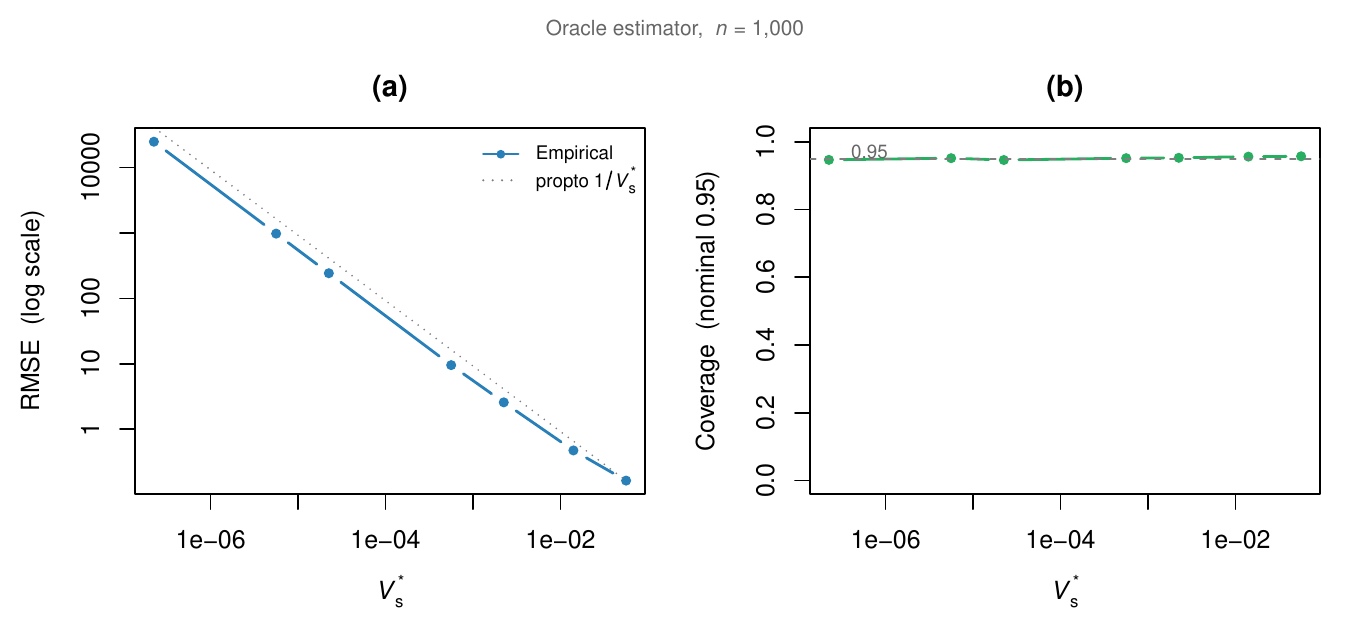}
\caption{Identification boundary. Left: empirical RMSE (solid) and the
theoretical $1/\Vs$ reference (dashed) on a log-log scale. Right:
CI coverage as $\Vs\to 0$. Oracle estimator, $n=1{,}000$.
Confirms Proposition~\ref{prop:failure} and Theorem~\ref{thm:clt}.}
\label{fig:boundary}
\end{figure}

\subsection{Calibration failure and the sensitivity bound}
\label{sec:mc-calibration}

\noindent
Proposition~\ref{prop:robust} characterises bias
under miscalibration and shows that the bound
$|\tau|\,\delta\,\E[|z|]/(2\Vs)$ is sharp over $\mathcal{H}_\delta$.
In these experiments $G$ is drawn from $p+\eta$ and the oracle subtracts
$m_0(X)=\mu(X)+\tau r(X)$, so the table and the figure evaluate parts (a)
and (b) of the proposition. The estimator that subtracts $\E[Y\mid X]$
obeys part (c).
Table~\ref{tab:sim3} and Figure~\ref{fig:bias} evaluate three calibration
error shapes at four values of $\delta$, with $n=2{,}000$ and $\tau=1$.
The simulated \emph{worst-case} shape is the bias-amplifying extremal
$\eta=\delta\,\mathrm{sgn}(2p-1)$, which the DGP truncates wherever
$p+\eta$ would exit $[0,1]$; per the proof of
Proposition~\ref{prop:robust}(b), its attainable bias is
$|\tau|\,\E[|z|\min(\delta,\min(p,1-p))]/(2\Vs)$ rather than the unclipped
bound. The tightness ratios in Table~\ref{tab:sim3}, falling from 0.99 to
0.86 as $\delta$ grows, match this truncation prediction to within Monte
Carlo error (the predicted ratios under the design's score distribution
are 0.99, 0.96, 0.92, and 0.86): the shortfall is truncation, not
sampling noise, and the admissible attenuating extremal
$-\delta\,\mathrm{sgn}(2p-1)$ attains the bound exactly. The
\emph{symmetric} shape
$\eta=\delta\sin(\pi p)$ satisfies $\E[(2p-1)\eta]\approx 0$ and produces
near-zero theoretical and empirical bias regardless of $\delta$, confirming
that calibration errors orthogonal to the signed score leave the estimator
unbiased. The \emph{linear} shape lies between these extremes.

\begin{table}[H]
\centering
\caption{Calibration failure and sharp sensitivity bound ($n=2{,}000$,
$\tau=1$, $\sigma_u=0.30$, $R=2{,}000$ replications).}
\label{tab:sim3}
\begin{threeparttable}
\begin{tabular}{llrrrr}
\toprule
$\eta$ shape & $\delta$ & Emp.\ bias & Theo.\ bias & Sharp bound & Tightness \\
\midrule
Worst-case
  & 0.05 & $+$0.434 & $+$0.439 & 0.439 & 0.989 \\
  & 0.10 & $+$0.852 & $+$0.879 & 0.879 & 0.970 \\
  & 0.15 & $+$1.199 & $+$1.318 & 1.318 & 0.910 \\
  & 0.20 & $+$1.515 & $+$1.756 & 1.756 & 0.862 \\ \midrule
Linear
  & 0.05 & $+$0.232 & $+$0.226 & 0.439 & 0.528 \\
  & 0.10 & $+$0.423 & $+$0.451 & 0.878 & 0.482 \\
  & 0.15 & $+$0.602 & $+$0.676 & 1.316 & 0.457 \\
  & 0.20 & $+$0.774 & $+$0.903 & 1.757 & 0.440 \\ \midrule
Symmetric
  & 0.05 & $-$0.009 & $\approx$0 & 0.439 & 0.020 \\
  & 0.10 & $-$0.003 & $\approx$0 & 0.878 & 0.004 \\
  & 0.15 & $+$0.000 & $\approx$0 & 1.317 & 0.000 \\
  & 0.20 & $+$0.008 & $\approx$0 & 1.757 & 0.004 \\
\bottomrule
\end{tabular}
\begin{tablenotes}\small
\item \textit{Notes}: Worst-case: $\eta=\delta\,\mathrm{sgn}(2p-1)$.
Linear: $\eta=\delta(2p-1)$. Symmetric: $\eta=\delta\sin(\pi p)$.
Theoretical bias: $B_{\mathrm{cal}}=\tau\E[(2p-1)\eta]/(2\Vs)$.
Sharp bound: $|\tau|\,\delta\,\E[|2p-1|]/(2\Vs)$.
Tightness $=|$emp.\ bias$|/$sharp bound. Oracle estimator.
The DGP truncates $p+\eta$ to $[0,1]$, so the worst-case shape is fully
attainable only where $\min(p,1-p)\geq\delta$; the resulting tightness
shortfall at larger $\delta$ is systematic (see the text).
The theoretical bias and the sharp bound are evaluated at the untruncated
$\eta$. The truncation explains why the empirical bias of the worst-case and
linear shapes falls below the theoretical bias at larger $\delta$.
\end{tablenotes}
\end{threeparttable}
\end{table}

\begin{figure}[H]
\centering
\includegraphics[width=\textwidth]{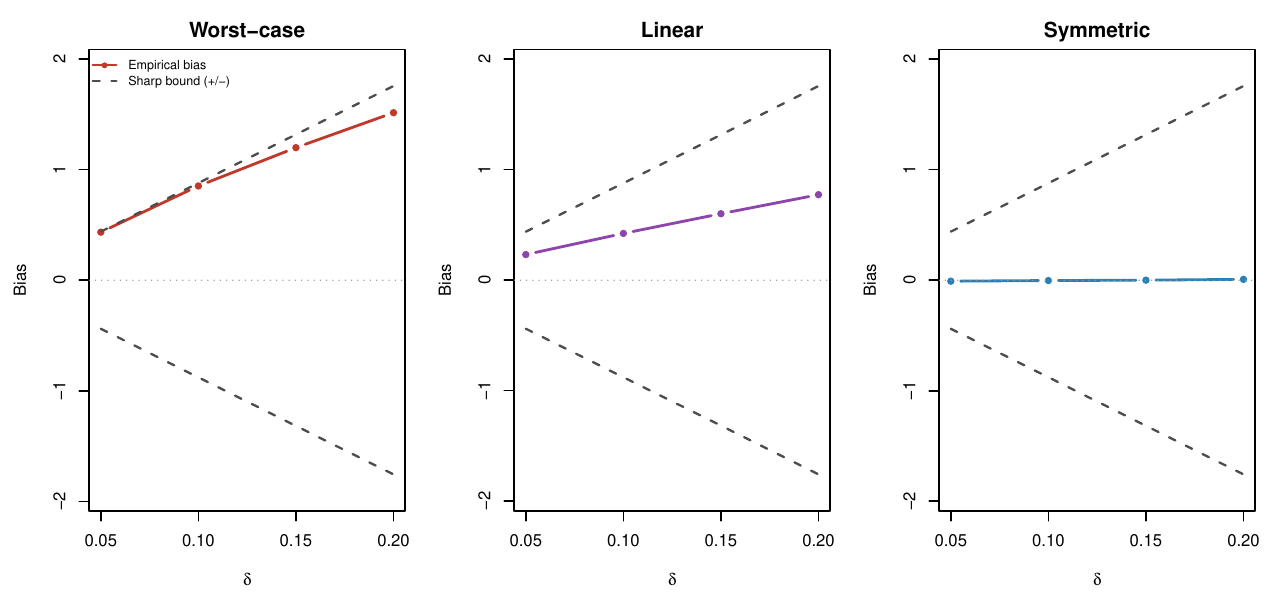}
\caption{Empirical bias (points) and sharp bound (dashed) as functions of
$\delta$ for three calibration error shapes. The worst-case shape nearly
attains the bound; the symmetric shape generates bias indistinguishable
from zero. Confirms Proposition~\ref{prop:robust}(a)--(b).}
\label{fig:bias}
\end{figure}

\subsection{Attenuation by hard-threshold classification}
\label{sec:mc-threshold}

\noindent
Table~\ref{tab:sim4} and Figure~\ref{fig:attenuation} evaluate the attenuation
result in a DGP satisfying $r(X)=\frac12$ a.s.\ and the conditional symmetry
condition of Appendix~\ref{app:threshold}, at $\sigma_u\in\{0.10,0.20,0.30\}$
with $n=1{,}000$ and $\tau=1$. The oracle and orthogonal estimators are
centred on $\tau=1$ in every setting. The threshold estimator converges to
approximately $\hat\kappa\tau$ and attenuation worsens sharply as $\sigma_u$
decreases: at $\sigma_u=0.10$, the threshold estimate is approximately $0.08$
where the truth is $1$.

\begin{table}[H]
\centering
\caption{Hard-threshold versus moment estimators ($n=1{,}000$, $\tau=1$,
$R=2{,}000$ replications).}
\label{tab:sim4}
\begin{threeparttable}
\begin{tabular}{rllrrr}
\toprule
$\sigma_u$ & $\hat\kappa$ & Estimator & Mean $\hat\tau$ & Bias & RMSE \\
\midrule
0.10 & 0.100 & Oracle    & 1.008 & $+$0.008 & 0.889 \\
     &       & Plug-in   & 0.897 & $-$0.103 & 0.719 \\
     &       & Threshold & 0.085 & $-$0.915 & 0.919 \\ \midrule
0.20 & 0.172 & Oracle    & 1.005 & $+$0.005 & 0.379 \\
     &       & Plug-in   & 0.944 & $-$0.056 & 0.353 \\
     &       & Threshold & 0.163 & $-$0.837 & 0.840 \\ \midrule
0.30 & 0.252 & Oracle    & 1.005 & $+$0.005 & 0.239 \\
     &       & Plug-in   & 0.952 & $-$0.048 & 0.232 \\
     &       & Threshold & 0.249 & $-$0.751 & 0.754 \\
\bottomrule
\end{tabular}
\begin{tablenotes}\small
\item \textit{Notes}: $\hat\kappa=2\,\overline{|p-\frac12|}$ is the
empirical attenuation factor. Under the conditions of
Appendix~\ref{app:threshold}, $\plim\,\hat\tau_{\mathrm{ht}}=\kappa\tau$.
Plug-in: degree-2 polynomial ridge models without cross-fitting.
\end{tablenotes}
\end{threeparttable}
\end{table}

\begin{figure}[H]
\centering
\includegraphics[width=\textwidth]{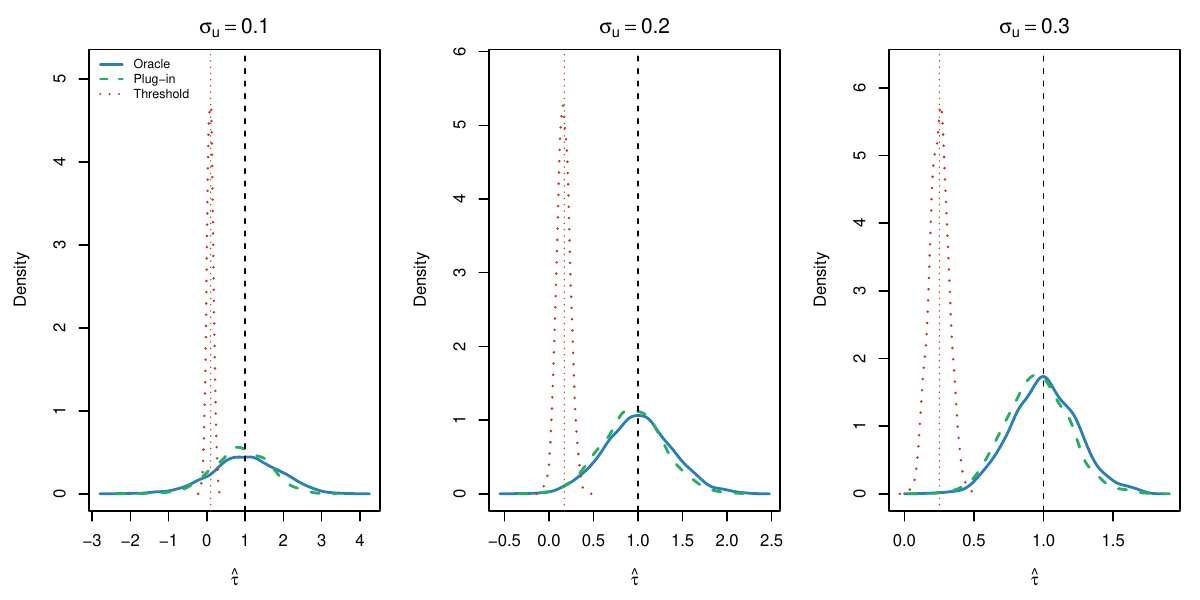}
\caption{Sampling distributions of the oracle, plug-in, and hard-threshold
estimators at three levels of score dispersion $\sigma_u$. The threshold
estimator is centred well below $\tau=1$ in each panel; attenuation worsens
as $\sigma_u$ decreases. Confirms Appendix~\ref{app:threshold}.}
\label{fig:attenuation}
\end{figure}

\subsection{Heterogeneous effects and the variance-weighted estimand}
\label{sec:mc-hetero}

\noindent
\begin{sloppypar}
Remark~\ref{rem:hetero} establishes that when the
structural effect varies with covariates, the moment estimator identifies the
variance-weighted average $\bar\tau=\E[\tau(X)\Var(p\mid X)]/\E[\Var(p\mid X)]$
rather than the simple mean $\E[\tau(X)]$. Table~\ref{tab:sim5} and
Figure~\ref{fig:weighted} confirm this in a DGP with
$\tau(X)=\tau_0+\tau_1 X_1$, $\tau_0=1$, $\tau_1=0.5$.
Design A holds $\Var(p\mid X)$ constant, so $\bar\tau=\tau_0=1.001$. Design B
introduces $X$-varying score variance $\Var(p\mid X)\propto e^{0.8X_1}$, which
upweights units with large positive $X_1$ and gives $\bar\tau=1.362\neq\E[\tau(X)]=1$.
In both designs, bias is negligible and shrinks towards zero as $n$ grows,
confirming that the oracle estimator correctly identifies the variance-weighted
estimand.
\end{sloppypar}

\begin{table}[H]
\centering
\caption{Heterogeneous effects: variance-weighted estimand recovery
($\tau(X)=\tau_0+\tau_1 X_1$, $\tau_0=1$, $\tau_1=0.5$,
$R=2{,}000$ replications).}
\label{tab:sim5}
\begin{threeparttable}
\begin{tabular}{llrrrr}
\toprule
Design & $n$ & $\bar\tau$ & Bias & RMSE & Coverage \\ \midrule
A: constant $\Var(p\mid X)$
  & 500   & 1.001 & $-$0.019 & 0.517 & 0.956 \\
  & 1,000 & 1.001 & $-$0.009 & 0.378 & 0.942 \\
  & 5,000 & 1.001 & $+$0.000 & 0.168 & 0.942 \\ \midrule
B: heterogeneous $\Var(p\mid X)$
  & 500   & 1.362 & $-$0.006 & 0.416 & 0.950 \\
  & 1,000 & 1.362 & $-$0.004 & 0.296 & 0.951 \\
  & 5,000 & 1.362 & $+$0.001 & 0.131 & 0.953 \\
\bottomrule
\end{tabular}
\begin{tablenotes}\small
\item \textit{Notes}: $\bar\tau=\E[\tau(X)\Var(p\mid X)]/\E[\Var(p\mid X)]$.
Design A: $\Var(p\mid X)$ constant; $\bar\tau=\tau_0=1$.
Design B: $\Var(p\mid X)=\sigma_u^2 e^{0.8X_1}$; $\bar\tau\neq\E[\tau(X)]=1$.
Oracle estimator with true $m(X)$ and $r(X)$.
\end{tablenotes}
\end{threeparttable}
\end{table}

\begin{figure}[H]
\centering
\includegraphics[width=.75\textwidth]{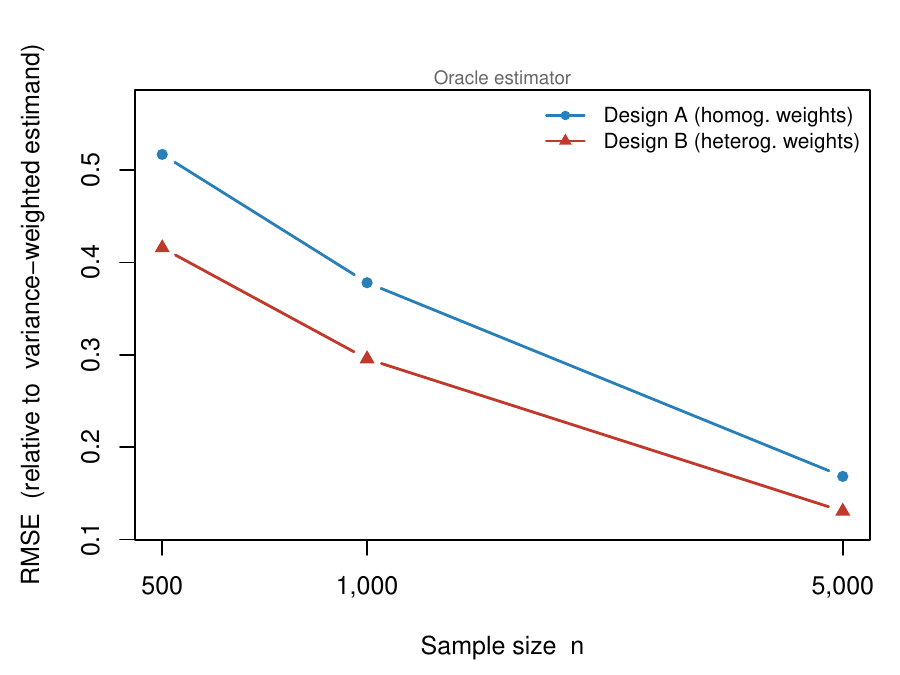}
\caption{RMSE relative to the variance-weighted estimand $\bar\tau$ as a
function of sample size, for Design A (constant weights, $\bar\tau=\tau_0=1$)
and Design B (heterogeneous weights, $\bar\tau=1.362$). Both designs converge
to their respective $\bar\tau$, confirming Remark~\ref{rem:hetero}.}
\label{fig:weighted}
\end{figure}

\section{Discussion}
\label{sec:discussion}
\noindent
We have established point identification of a structural latent-group
coefficient $\tau$, a sharp characterisation of identification failure,
oracle inference and plug-in consistency results, and robustness
to calibration failure. This section discusses three further topics: the
attenuation induced by hard-threshold classification, the interpretation of
the estimand under heterogeneous effects, and the division of labour with
the companion papers together with the directions that remain open.

A common alternative to the moment estimator~\eqref{eq:oracle} is to threshold
the score at $p=\frac12$, form a binary indicator $\tilde G:=\1\{p>\frac12\}$,
and estimate the group gap as the difference in conditional means across the two
induced cells. Appendix~\ref{app:threshold} shows that this estimator converges
to $\kappa\tau$ with $\kappa=2\E[|p-\frac12|]\in(0,1)$ under the conditions
of Proposition~\ref{prop:threshold}, among them $r(X)=\frac12$ a.s.\ and a
conditionally symmetric score, so for $\tau\neq 0$ it is inconsistent whenever
classification is imperfect, while the moment estimator is consistent.
The Monte Carlo evidence in Section~\ref{sec:mc-threshold} confirms that
attenuation can be severe: when score dispersion is low the threshold estimator
recovers less than ten percent of the true coefficient.

Under heterogeneous effects, Remark~\ref{rem:hetero} shows that the moment
equation identifies the variance-weighted average
$\bar\tau=\E[\tau(X)\Var(p\mid X)]/\E[\Var(p\mid X)]$, with weight equal to
the local informativeness of the score; Section~\ref{sec:mc-hetero}
verifies the weighting formula, including a design where $\bar\tau$ and
$\E[\tau(X)]$ are held apart by construction.

A companion applied paper \citep{kurbucz2026kbs} gives the attenuation
factor of the moment estimator computed with a thresholded label, for any
confidence threshold and without the conditions $r(X)=\frac12$ and
conditional symmetry. Together with the collapse result, this factor is the
basis of a diagnostic for pseudo-labelled regression that compares the
moment, hard-label and supervised estimators by their estimated mean squared
errors. The diagnostic includes the conditional-calibration test battery and
the labelled-subsample recalibration protocol referred to in
Remark~\ref{rem:obtain}. The extension to several latent classes, in which
identification is governed by the spectrum of the residual score covariance
matrix and coarsening the probability vector distorts the coefficient vector,
is developed in \citet{kurbucz2026coarsening}. For a binary group, its
sensitivity bound reduces to Proposition~\ref{prop:robust}(c).

Three directions remain open. First, the bound of
Proposition~\ref{prop:robust}(b) is sharp over the unrestricted class
$\mathcal{H}_\delta$, and that of part (c) whenever $\min(p,1-p)\geq\delta$,
but tighter bounds should be achievable under shape restrictions on $\eta(p,X)$ (monotonicity,
smoothness, or orthogonality to a known basis). Second, inference that is
uniformly valid as $\Vs\to 0$---local to the identification boundary---is
not covered by any of the results above. Third, the framework treats the
score-generating mechanism as externally supplied; jointly modelling score
construction and downstream inference, including the sampling uncertainty
of a first-stage classifier, is an open problem.

\section{Conclusion}
\label{sec:conclusion}
\noindent
This paper has developed a framework for identifying and estimating a
structural group effect when the binary group indicator is latent but a
calibrated probability score is observed. Under a constant-coefficient
conditional mean model and the calibration condition $\E[G\mid p,X]=p$, the
structural coefficient $\tau$ is point-identified by a closed-form ratio of
observable moments, provided the score carries residual variation beyond
covariates. Identification fails precisely when this residual variation is
absent, and the failure is characterised constructively by an explicit
family of observationally equivalent models---outcome-coupled latent
indicators---whose coefficients span a nondegenerate interval.

Several conclusions follow from the analysis. The identified coefficient is a
within-covariate-cell structural effect, distinct from the marginal group mean
gap---itself identified in closed form under conditional calibration; the two
coincide if and only if the latent groups are covariate-balanced.
The oracle estimator is $\sqrt{n}$-consistent and asymptotically normal with a
closed-form sandwich variance and, unlike hard-threshold classification, is
not attenuated when the score is imperfectly concentrated.
When calibration is imperfect, the bias admits an exact formula and is bounded
by a sharp sensitivity bound that scales inversely with the residual score
variance, consistently with the identification result.

The Monte Carlo evidence confirms each of these predictions quantitatively:
the oracle estimator is approximately unbiased and asymptotically normal, RMSE
diverges at the predicted rate as $\Vs\to 0$, calibration errors produce bias
bounded by the sharp formula, and hard-threshold classification induces the
predicted attenuation factor $\kappa$.

Feasible inference with flexible nuisance estimators follows from the
orthogonal score of Appendix~\ref{app:ortho} and the double machine
learning framework of \citet{chernozhukov2018}. The practical entry point
for applied work is equally concrete: conditional calibration is attainable
by construction for posterior-type scores and testable-and-enforceable on a
labelled validation subsample (Remark~\ref{rem:obtain}), with the companion
diagnostic paper \citep{kurbucz2026kbs} supplying the operational protocol.
The framework developed here---centred on a calibrated probability as a
proxy for latent membership---has natural applications in fairness auditing
with proxy attributes such as BISG, in distributional analysis, and in any
empirical setting where group indicators are administratively missing but
predictable from observed characteristics.

\clearpage
\bibliographystyle{elsarticle-harv}
\bibliography{references}

\appendix
\renewcommand{\thesection}{A}
\renewcommand{\thesubsection}{A.\arabic{subsection}}
\setcounter{theorem}{0}
\renewcommand{\thetheorem}{A.\arabic{theorem}}

\section{Appendix}
\label{app:proofs}

\noindent
Throughout the Appendix we use without further notice: $z=2p-1$, $R=Y-m(X)$,
$a=p-r(X)$, $\Vs=\E[a^{2}]$, and the identities
$\E[R\mid X]=\E[a\mid X]=0$ from~\eqref{eq:derived}.

\subsection{Proof of Lemma~\ref{lem:basic}}
\label{app:lem-basic}

\noindent
\textit{First identity.}
By the tower property and Assumption~\ref{ass:cal},
\[
  \pi(X)=\E[G\mid X]=\E[\E[G\mid p,X]\mid X]=\E[p\mid X]=r(X).
\]
\textit{Second identity.}
By the tower property and Assumption~\ref{ass:mean},
$m(X)=\E[\mu(X)+\tau G\mid X]=\mu(X)+\tau\pi(X)=\mu(X)+\tau r(X)$,
where the last equality uses the first identity.\hfill$\square$

\subsection{Proof of Lemma~\ref{lem:residual}}
\label{app:lem-resid}

\noindent
$\E[R\mid G,p,X]=\E[Y\mid G,p,X]-m(X)=(\mu(X)+\tau G)-(\mu(X)+\tau r(X))
=\tau(G-r(X))$, using Assumption~\ref{ass:mean} and the second identity
of Lemma~\ref{lem:basic}.
Define $\varepsilon:=R-\tau(G-r(X))$; then $\E[\varepsilon\mid G,p,X]=0$
by construction.\hfill$\square$

\subsection{Proof of Theorem~\ref{thm:moment}}
\label{app:thm-moment}

\noindent
We prove $\E[zR]=2\tau\Vs$.

\smallskip
\noindent\textit{Step 1 (noise term).}
By Lemma~\ref{lem:residual}, $R=\tau(G-r(X))+\varepsilon$. Since $z$ is
$\sigma(p,X)$-measurable and $\E[\varepsilon\mid G,p,X]=0$,
\begin{align*}
  \E[zR]
    &= \tau\E[z(G-r(X))] + \E[z\varepsilon] \\
    &= \tau\E[z(G-r(X))]
       + \E\!\left[z\,\E[\varepsilon\mid G,p,X]\right] \\
    &= \tau\E[z(G-r(X))].
\end{align*}

\noindent\textit{Step 2 (calibration).}
Conditioning on $(p,X)$ and applying Assumption~\ref{ass:cal},
\[
  \E[z(G-r(X))]=\E[(2p-1)\E[G-r(X)\mid p,X]]=\E[(2p-1)(p-r(X))].
\]
\noindent\textit{Step 3 (algebra).}
Write $a=p-r(X)$ so $2p-1=2a+(2r(X)-1)$. Then
$(2p-1)(p-r(X))=2a^{2}+(2r(X)-1)a$.
Since $\E[a\mid X]=0$, taking expectations:
$\E[(2r(X)-1)a]=0$, so $\E[(2p-1)(p-r(X))]=2\E[a^{2}]=2\Vs$.
Combining: $\E[zR]=2\tau\Vs$, which is part (a).

\smallskip
\noindent\textit{Part (b).}
Under Assumption~\ref{ass:var}, $\Vs>0$, so dividing both sides of
\eqref{eq:moment} by $2\Vs$ gives \eqref{eq:tau-id}. Every quantity on the
right-hand side of \eqref{eq:tau-id} is a functional of the joint law of
$(Y,X,p)$, so $\tau$ is point-identified.\hfill$\square$

\subsection{Proof of Proposition~\ref{prop:failure}}
\label{app:prop-failure}

\noindent
\textit{Part (a).}  If $\Vs=0$ then $a=0$ a.s., so $z=2r(X)-1$ is
$\sigma(X)$-measurable and $\E[zR]=\E[z\E[R\mid X]]=0$. The right-hand
side $2\tau\Vs=0$.

\textit{Part (b).}
Since $\Vs=0$, part (c) gives $p=r(X)$ a.s., so $\sigma(p,X)=\sigma(X)$;
in particular Assumption~\ref{ass:cal} for a candidate indicator $G'$
reduces to $\E[G'\mid X]=r(X)$ a.s.

Fix $\tau'$ with $|\tau'|\leq\bar\tau=2\underline r\,\underline\kappa$.
Write $u:=Y-m(X)$ and define
\begin{gather*}
  \psi:=g(u)-\E[g(u)\mid X],\qquad
  \kappa(X):=\E[u\,g(u)\mid X],\\
  c(X):=\frac{\tau'\,r(X)\bigl(1-r(X)\bigr)}{\kappa(X)} .
\end{gather*}
Note $u\,g(u)=u^{2}/(1+|u|)\geq 0$, and by hypothesis
$\kappa(X)\geq\underline\kappa>0$ a.s., so $c$ is well defined. Since
$|g|<1$ we have $|\psi|<2$, and since $r(1-r)\leq\tfrac14$,
\begin{gather*}
  |c(X)|\;\leq\;\frac{|\tau'|}{4\underline\kappa}
  \;\leq\;\frac{\bar\tau}{4\underline\kappa}
  \;=\;\frac{\underline r}{2},\\
  \text{hence}\qquad
  |c(X)\psi|<\underline r\;\leq\;\min\bigl(r(X),\,1-r(X)\bigr)\as
\end{gather*}
Therefore
\[
  q(Y,X):=r(X)+c(X)\,\psi\;\in\;[0,1]\as
\]
Extend the probability space with $V\sim\mathrm{Uniform}(0,1)$ independent
of $(Y,X,p)$ and set $G':=\1\{V\leq q(Y,X)\}$, so that
$\Pbb(G'=1\mid Y,X)=q(Y,X)$. The observable triple $(Y,X,p)$ is untouched
by the construction, so its joint distribution is exactly the original
one: observational equivalence is exact, not merely at the level of
$m(X)$.

\emph{Assumption~\ref{ass:cal} for $G'$:}
$\E[G'\mid X]=\E[q(Y,X)\mid X]=r(X)$ because $\E[\psi\mid X]=0$; by the
reduction above this is Assumption~\ref{ass:cal}.

\emph{Assumption~\ref{ass:mean} for $(G',\tau',\mu')$ with
$\mu'(X):=m(X)-\tau'r(X)$:}
first, $\E[Y\psi\mid X]=\E[u\,g(u)\mid X]=\kappa(X)$, using
$\E[u\mid X]=0$ and $\E[\psi\mid X]=0$. For the binary $G'$ the conditional
mean given $(G',X)$ is determined by its two branch values
$\E[Y\mid G'=g,X]:=\E[Y\,\1\{G'=g\}\mid X]\,/\,\Pbb(G'=g\mid X)$, both well
defined since $\min\bigl(r(X),1-r(X)\bigr)\geq\underline r>0$ a.s. Hence,
for a.e.\ $x$:
\begin{align*}
  \E[Y\mid G'=1,X=x]
  &=\frac{\E[Y\,q(Y,X)\mid X=x]}{\E[q(Y,X)\mid X=x]}
   =\frac{r(x)m(x)+c(x)\kappa(x)}{r(x)}\\
  &=m(x)+\tau'\bigl(1-r(x)\bigr),\\[2pt]
  \E[Y\mid G'=0,X=x]
  &=\frac{m(x)-r(x)m(x)-c(x)\kappa(x)}{1-r(x)}
   =m(x)-\tau'r(x),
\end{align*}
where both final equalities substitute
$c(x)\kappa(x)=\tau'r(x)(1-r(x))$. Combining the two cases,
$\E[Y\mid G',X]=\mu'(X)+\tau'G'$ a.s., and since
$\sigma(G',p,X)=\sigma(G',X)$, Assumption~\ref{ass:mean} holds with
coefficient $\tau'$. Assumption~\ref{ass:moments} concerns only $(Y,p)$
and is unchanged.

Every $\tau'\in[-\bar\tau,\bar\tau]$ is therefore the structural
coefficient of a model observationally indistinguishable from the
original, so $\tau$ is not point-identified.

\textit{Part (c).}
$\Vs=\E[\Var(p\mid X)]=0\iff\Var(p\mid X)=0\as\iff p=r(X)\as$\hfill$\square$

\subsection{Proofs of Theorem~\ref{thm:clt} and Corollary~\ref{cor:wald}}
\label{app:inference}

\noindent
\textit{Proof of Theorem~\ref{thm:clt}.}
Write $U_i:=(2p_i-1)(Y_i-m(X_i))$ and $V_i:=(p_i-r(X_i))^{2}$, so
$\tauor=\bar U_n/(2\bar V_n)$. By the WLLN, $\bar V_n\pto\Vs>0$. By the
bivariate CLT (valid under Assumption~\ref{ass:moments}),
$\sqrt{n}(\bar U_n-\E[U_i],\,\bar V_n-\Vs)\dto\N(0,\Sigma)$ where
$\Sigma=\Var(U_i,V_i)$.

Apply the delta method to $f(u,v)=u/(2v)$ with gradient
\[
  \nabla f\big|_{(\E[U_i],\,\Vs)}
  =\left(\frac{1}{2\Vs},\;-\frac{\E[U_i]}{2(\Vs)^{2}}\right)^{\!\top}
  =\left(\frac{1}{2\Vs},\;-\frac{\tau}{\Vs}\right)^{\!\top}.
\]
Setting $W_i:=U_i-2\tau V_i=\psi_i$ (the centred score), expanding
$\Var(U_i)$, $\Cov(U_i,V_i)$, and $\Var(V_i)$ in terms of $W_i$ and $V_i$,
and collecting: all cross-terms cancel and
$(\nabla f)^{\top}\Sigma\,\nabla f=\E[W_i^{2}]/(2\Vs)^{2}=\E[\psi_i^{2}]/(2\Vs)^{2}$.\hfill$\square$

\textit{Proof of Corollary~\ref{cor:wald}.}
$\frac{1}{n}\sum_i\hat\psi_i^{2}\pto\E[\psi_i^{2}]$ by WLLN and
$\tauor\pto\tau$; the denominator converges to $(2\Vs)^{2}$. Slutsky gives
$\hat\sigma^{2}_{\mathrm{or}}\pto\sigma^{2}_{\mathrm{or}}$, and coverage
follows from Theorem~\ref{thm:clt}.\hfill$\square$

\subsection{Proofs of Lemma~\ref{lem:denom} and Proposition~\ref{prop:plugin}}
\label{app:plugin}

\noindent
\textit{Proof of Lemma~\ref{lem:denom}.}
Let $\delta_i:=\hat r(X_i)-r(X_i)$. Then
$(p_i-\hat r(X_i))^{2}-(p_i-r(X_i))^{2}=\delta_i(\hat r(X_i)+r(X_i)-2p_i)$.
Since all three terms lie in $[0,1]$, the second factor is bounded by $2$ in
absolute value, giving
$n^{-1}\sum_i|\ldots|\leq 2n^{-1}\sum_i|\delta_i|
\leq 2(n^{-1}\sum_i\delta_i^{2})^{1/2}\pto 0$
by Cauchy--Schwarz and the hypothesis.\hfill$\square$

\textit{Proof of Proposition~\ref{prop:plugin}.}
\emph{Numerator:}
write $\frac{1}{n}\sum_i(2p_i-1)(Y_i-\hat m(X_i))=A_n-B_n$ where
$A_n\pto 2\tau\Vs$ (WLLN) and $|B_n|\leq(n^{-1}\sum_i(\hat m-m)^{2})^{1/2}
\pto 0$ by Cauchy--Schwarz.

\emph{Denominator:} converges to $\Vs>0$ by WLLN and Lemma~\ref{lem:denom}.

\emph{Conclusion:} CMT gives $\hat\tau\pto\tau$.\hfill$\square$

\subsection{Proof of Proposition~\ref{prop:robust}}
\label{app:robust}

\noindent
\textit{Part (a).}
Under \eqref{eq:miscal-model}, Assumption~\ref{ass:mean} gives
$Y=\mu(X)+\tau G+\varepsilon$ with $\E[\varepsilon\mid G,p,X]=0$, so
$Y-m_0(X)=\tau(G-r(X))+\varepsilon$. Step 2 of the proof of
Theorem~\ref{thm:moment}, with $\E[G\mid p,X]=p+\eta$ in place of $p$, gives
$\E[z(G-r(X))]=\E[(2p-1)(p+\eta-r(X))]=\E[(2p-1)(\eta+a)]
=\E[z\eta]+2\Vs$.
Hence $\E[z(Y-m_0(X))]=\tau(\E[z\eta]+2\Vs)$, and
$\plim\,\tauorz=\tau+\tau\E[z\eta]/(2\Vs)=\tau+B_{\mathrm{cal}}$.\hfill$\square$

\textit{Part (b).}
From~\eqref{eq:bias}, $|B_{\mathrm{cal}}|=|\tau||\E[z\eta]|/(2\Vs)$.
For $\eta\in\mathcal{H}_\delta$, H\"{o}lder gives
$|\E[z\eta]|\leq\E[|z||\eta|]\leq\delta\E[|z|]$, so
$|B_{\mathrm{cal}}|\leq|\tau|\delta\E[|z|]/(2\Vs)$.

For attainment, consider $\eta^{*}=-\delta\,\sgn(z)$. Then
$\E[z\eta^{*}]=-\delta\E[|z|]$, so $|B_{\mathrm{cal}}|$ equals the bound.
Admissibility: $|\eta^{*}|\leq\delta$; and $p+\eta^{*}\in[0,1]$ a.s.\
because on $\{z>0\}$ (i.e.\ $p>\tfrac12$) we have
$p-\delta\geq\tfrac12-\delta\geq 0$ and $p-\delta\leq 1$, on $\{z<0\}$ we
have $p+\delta\leq\tfrac12+\delta\leq 1$ and $p+\delta\geq 0$, and on
$\{z=0\}$, $\eta^{*}=0$. Hence $\eta^{*}\in\mathcal{H}_\delta$ and the
supremum is attained.

Two remarks on the geometry of the extremum. The sign-reversed shape
$+\delta\,\sgn(z)$, which also achieves $|\E[z\eta]|=\delta\E[|z|]$
algebraically, is admissible only when $\min(p,1-p)\geq\delta$ a.s.: near
the endpoints of the score distribution the constraint $p+\eta\in[0,1]$
truncates it, and the worst admissible \emph{bias-amplifying} error is
$\sgn(z)\min\bigl(\delta,\min(p,1-p)\bigr)$, with attained bias
$|\tau|\,\E\bigl[|z|\min(\delta,\min(p,1-p))\bigr]/(2\Vs)$, strictly below
the bound when the score places mass near $\{0,1\}$. The worst-case
\emph{magnitude} over $\mathcal{H}_\delta$ is nevertheless attained
exactly, in the attenuating direction, as shown above.\hfill$\square$

\textit{Part (c).}
Under \eqref{eq:miscal-model},
$\E[Y\mid X]=\mu(X)+\tau\E[G\mid X]=\mu(X)+\tau\bigl(r(X)+\bar\eta(X)\bigr)
=m_0(X)+\tau\bar\eta(X)$, so $R=Y-m(X)=\tau(G-r(X)-\bar\eta(X))+\varepsilon$.
As in part (a), $\E[z(G-r(X))]=\E[z\eta]+2\Vs$. Writing
$z=2a+(2r(X)-1)$ and using $\E[a\mid X]=0$ and
$\E[(2r(X)-1)\bar\eta(X)]=\E[(2r(X)-1)\eta]$ gives
$\E[z\bar\eta(X)]=\E[(2r(X)-1)\eta]$. Hence
\[
  \E[zR]=\tau\bigl(2\Vs+\E[z\eta]-\E[(2r(X)-1)\eta]\bigr)
        =2\tau\Vs+2\tau\E[a\eta],
\]
and since the denominator of~\eqref{eq:oracle} converges to $2\Vs$,
$\plim\,\tauor=\tau+\tau\E[a\eta]/\Vs$. The limit in part (a) minus this
limit is $\tau\E[(2r(X)-1)\eta]/(2\Vs)$, which vanishes when
$\E[\eta\mid X]=0$ a.s.\ or $r(X)=\frac12$ a.s. For
$\eta\in\mathcal{H}_\delta$, H\"{o}lder gives
$|\E[a\eta]|\leq\delta\E[|a|]$. If $\min(p,1-p)\geq\delta$ a.s., the shape
$\eta=-\delta\,\sgn(a)$ satisfies $|\eta|\leq\delta$ and
$p+\eta\in[p-\delta,p+\delta]\subseteq[0,1]$, so it lies in
$\mathcal{H}_\delta$, and $\E[a\eta]=-\delta\E[|a|]$ attains the
bound.\hfill$\square$

\subsection{Attenuation by hard-threshold classification}
\label{app:threshold}

\noindent
Define $\tauht:=\bar R_{\{p>1/2\}}-\bar R_{\{p\leq 1/2\}}$ where
$\bar R_A:=|A|^{-1}\sum_{i\in A}R_i$.

\begin{aproposition}[Attenuation]
\label{prop:threshold}
Under Assumptions~\ref{ass:mean}--\ref{ass:var}, with $r(X)=\frac12$
a.s., $\Pbb(p\in(0,1))>0$, $\Pbb(p=\frac12)=0$, and the condition that for
$\Pbb_X$-almost every $x$ the conditional law of $a=p-\frac12$ given $X=x$
is symmetric around zero,
\begin{equation}
  \plim\;\tauht = \kappa\tau,\qquad
  \kappa := 2\E[|a|] = 2\E\!\left[|p-\tfrac12|\right]\;\in\;(0,1).
  \label{eq:attenuation}
\end{equation}
\end{aproposition}

\begin{proof}
By Lemma~\ref{lem:residual} and Assumption~\ref{ass:cal},
$\E[R\mid p,X]=\tau(p-r(X))=\tau a$ where $a=p-\frac12$ under
$r(X)=\frac12$ a.s.

\textit{Computing $\E[R\mid p>\frac12]$:}
since $\{p>\frac12\}=\{a>0\}$, the tower property gives
\[
  \E[R\mid p>\tfrac12]
  = \tau\,\E[a\mid a>0]
  = \tau\,\frac{\E[a\,\1\{a>0\}]}{\Pbb(a>0)}.
\]
Conditional on $X=x$, the symmetry condition and $\Pbb(a=0)=0$ give
$\Pbb(a>0\mid X=x)=\frac12$ for $\Pbb_X$-almost every $x$ and, since $a\1\{a>0\}=\frac12|a|+\frac12 a$ with $\E[a\mid X=x]=0$,
also $\E[a\,\1\{a>0\}\mid X=x]=\frac12\E[|a|\mid X=x]$.
Integrating over $X$ gives $\Pbb(a>0)=\frac12$ and
$\E[a\,\1\{a>0\}]=\frac12\E[|a|]$. Hence
\[
  \E[R\mid p>\tfrac12]=\tau\cdot\frac{\tfrac12\E[|a|]}{\tfrac12}=\tau\E[|a|].
\]
By the same argument applied to $\{a\leq 0\}$, $\E[R\mid p\leq\frac12]=-\tau\E[|a|]$.
Both conditioning events have probability $\frac12$, so the WLLN gives $\plim\,\tauht=2\tau\E[|a|]=\kappa\tau$.

\textit{$\kappa\in(0,1)$:}
$a\in[-\frac12,\frac12]$ under $r(X)=\frac12$ a.s., so $\kappa\leq 1$
with equality iff $p\in\{0,1\}$ a.s., which is excluded.
Positivity follows from $\Vs>0$.
\end{proof}

The condition that $a\mid X$ is symmetric around zero is used at two steps,
together with $\Pbb(p=\frac12)=0$ in the first:
to obtain $\Pbb(a>0)=\frac12$ and $\E[a\,\1\{a>0\}]=\frac12\E[|a|]$,
which together give $\E[a\mid a>0]=\E[|a|]$.
Without the conditional symmetry assumption, the conditional means on the
two thresholded cells no longer simplify to $\pm\tau\E[|a|]$, so the
attenuation factor generally lacks the closed form $\kappa=2\E[|a|]$.
The condition that $r(X)=\frac12$ a.s.\ ensures that the threshold
$p=\frac12$ coincides with the zero of the score residual $a=p-r(X)$, so
that $\{p>\frac12\}=\{a>0\}$; without it the probability limit of $\tauht$
generally no longer admits the simple closed form $\kappa\tau$ with
$\kappa=2\E[|a|]$.

\subsection{Covariance identity and IV interpretation}
\label{app:iv}

\noindent
\begin{alemma}[Covariance identity]
\label{lem:cov-gp}
Under Assumption~\ref{ass:cal}, $\Cov(G,p\mid X)=\Var(p\mid X)$ a.s.
\end{alemma}

\begin{proof}
$\E[Gp\mid X]=\E[p\E[G\mid p,X]\mid X]=\E[p^{2}\mid X]$ by
Assumption~\ref{ass:cal}. Also $\E[G\mid X]=\pi(X)=r(X)=\E[p\mid X]$ by
Lemma~\ref{lem:basic}. Hence
$\Cov(G,p\mid X)=\E[p^{2}\mid X]-(\E[p\mid X])^{2}=\Var(p\mid X)$.
\end{proof}

Lemma~\ref{lem:cov-gp} makes the IV analogy precise. We compute the
slope of the population regression of $G-r(X)$ on $a=p-r(X)$:
\begin{align*}
  \Cov(G-r(X),\,a)
  &= \Cov(G,p) - \Cov(r(X),p) \\
  &= \E[\Cov(G,p\mid X)] \\
  &\quad + \Cov(\E[G\mid X],\E[p\mid X])
     - \Cov(r(X),r(X)) \\
  &= \E[\Var(p\mid X)] + \Var(r(X)) - \Var(r(X)) \\
  &= \Vs.
\end{align*}
where the second equality uses the law of total covariance and the third uses
Lemma~\ref{lem:cov-gp} and $\E[G\mid X]=\E[p\mid X]=r(X)$. Since
$\Var(a)=\Vs$, the regression slope is $\Cov(G-r(X),a)/\Var(a)=1$.

The population regression of $R$ on $G-r(X)$ has slope $\tau$ by
Lemma~\ref{lem:residual}. Formula~\eqref{eq:tau-id} is therefore the IV
estimand from regressing $R$ on $(G-r(X))$ with instrument $a$: the ratio of
the reduced-form slope $\Cov(R,a)/\Vs$ to the first-stage slope $1$.

\subsection{Orthogonality analysis and the orthogonal score}
\label{app:ortho}

\subsubsection*{Gateaux derivatives of the original score}

\noindent
The score $\psi(W;\tau,m,r):=(2p-1)(Y-m(X))-2\tau(p-r(X))^{2}$ has two
Gateaux derivatives.

\medskip\noindent\textit{$m$-derivative in direction $\delta_m$:}
\[
  \frac{d}{dh}\E[\psi(W;\tau,m+h\delta_m,r)]\Big|_{h=0}
  = -\E[(2p-1)\delta_m(X)]
  = -\E[(2r(X)-1)\delta_m(X)],
\]
where the last equality uses $\E[2p-1\mid X]=2r(X)-1$. This is non-zero in
general, vanishing for all $\delta_m$ only when $r(X)=\frac12$ a.s.

\medskip\noindent\textit{$r$-derivative in direction $\delta_r$:}
\begin{align*}
  \frac{d}{dh}\E[\psi(W;\tau,m,r+h\delta_r)]\Big|_{h=0}
  &= 4\tau\E[(p-r(X))\delta_r(X)] \\
  &= 4\tau\E\!\left[\E[a\mid X]\,\delta_r(X)\right] \\
  &= 0.
\end{align*}
since $\E[a\mid X]=\E[p-r(X)\mid X]=0$ exactly. The $r$-derivative therefore
vanishes automatically; only the $m$-direction fails orthogonality.

\subsubsection*{A Neyman-orthogonal reformulation}

\noindent
Replacing $(2p-1)$ by $2(p-r(X))=2a$ in the numerator gives the score of
\eqref{eq:ort-score},
\[
  \tilde\psi(W;\tau,m,r) := 2(p-r(X))\bigl(Y-m(X)-\tau(p-r(X))\bigr) = 2a(R-\tau a).
\]

\medskip\noindent\textit{$m$-derivative:}
$\frac{d}{dh}\E[\tilde\psi(W;\tau,m+h\delta_m,r)]\big|_{h=0}
= -2\E[(p-r(X))\delta_m(X)] = -2\E[\E[a\mid X]\,\delta_m(X)] = 0$.

\medskip\noindent\textit{$r$-derivative:}
With $a_h=p-(r+h\delta_r)(X)$,
\[
  \tilde\psi(W;\tau,m,r+h\delta_r)=2(a-h\delta_r(X))(R-\tau a+\tau h\delta_r(X)).
\]
Differentiating at $h=0$:
$-2\E[\delta_r(X)(R-\tau a)]+2\tau\E[a\,\delta_r(X)]$.
The first term vanishes because $\E[R-\tau a\mid X]=\E[R\mid X]-\tau\E[a\mid X]=0$.
The second term vanishes because $\E[a\mid X]=0$. Hence the $r$-derivative is
also zero, and $\tilde\psi$ is Neyman-orthogonal in both directions; the natural
estimator it defines is~\eqref{eq:ort-estimator}.

\subsubsection*{What the orthogonal score does and does not give}

\noindent
When nuisances are known, both $\hat\tau$ and $\hat\tau_{\mathrm{ort}}$
converge to $\tau$ and are asymptotically equivalent, since
$n^{-1}\sum_i(2p_i-1)R_i = 2n^{-1}\sum_i a_iR_i+O_P(n^{-1/2})$
(the cross-term $n^{-1}\sum_i(2r(X_i)-1)R_i$ is $O_P(n^{-1/2})$ by the CLT
and $\E[(2r(X)-1)R]=0$).

\begin{sloppypar}
With estimated nuisances, the situation differs. The cross-term becomes
$n^{-1}\sum_i(2\hat r(X_i)-1)(Y_i-\hat m(X_i))$, which is
$O_P(\|\hat r-r\|_{L^2})$ --- not $o_P(n^{-1/2})$ unless $\hat r$ converges
at rate $o(n^{-1/2})$. Consequently, with estimated nuisances $\hat\tau$ and
$\hat\tau_{\mathrm{ort}}$ are in general \emph{not} asymptotically equivalent.
The orthogonal score places $\hat\tau_{\mathrm{ort}}$ within the standard DML
framework of \citet{chernozhukov2018}, suggesting that $\sqrt{n}$-normality
under cross-fitting should hold under appropriate regularity conditions; a
formal proof is left to subsequent work.
\end{sloppypar}

\end{document}